\documentclass[10pt]{article}

\usepackage{amsmath,amssymb,epsf,epsfig,amsthm,times,ifthen,epic,psfrag,color}
\usepackage{printtime}
\usepackage{enumerate}
\usepackage{cite}    
\usepackage{fancyhdr}
\usepackage{setspace} 
\usepackage{lastpage} 

\def\submitteddate{April 18, 2010}
\def\reviseddate{August 9, 2010}

\begin{document}

\newcommand{\creationtime}{\today \ \ @ \theampmtime}

\pagestyle{fancy}
\renewcommand{\headrulewidth}{0cm}
\chead{\footnotesize{Appuswamy-Franceschetti-Karamchandani-Zeger}}
\rhead{\footnotesize{\reviseddate}}
\lhead{}
\cfoot{Page \arabic{page} of \pageref{LastPage}} 

\renewcommand{\qedsymbol}{$\blacksquare$} 


\newtheorem{theorem}              {Theorem}     [section]
\newtheorem{lemma}      [theorem] {Lemma}
\newtheorem{corollary}  [theorem] {Corollary}
\newtheorem{proposition}[theorem] {Proposition}

\theoremstyle{remark}
\newtheorem{algorithm}  [theorem] {Algorithm}
\newtheorem{conjecture} [theorem] {Conjecture}

\theoremstyle{definition}         
\newtheorem{remark}     [theorem] {Remark}
\newtheorem{definition} [theorem] {Definition}
\newtheorem{example}    [theorem] {Example}
\newtheorem*{claim}  {Claim}
\newtheorem*{notation}  {Notation}

\newcommand{\Comment}[1]{& [\mbox{from  #1}]}
\newcommand{\mc}{\mathcal}
\newcommand{\mb}{\mathbf}
\newcommand{\abs}[1]{\left\lvert #1 \right\rvert}
\newcommand{\card}[1]{\abs{#1}}
\newcommand{\Network}{\mathcal{N}}
\newcommand{\node}{v}
\newcommand{\edge}{e}
\newcommand{\nodes}{\mathcal{V}}
\newcommand{\vertices}[1]{\nodes}
\newcommand{\inNodes}[1]{V_{#1}}
\newcommand{\Capacity}[1]{w(#1)}
\newcommand{\edges}{\mathcal{E}}
\newcommand{\inEdges}[1]{\mathcal{E}_i(#1)}
\newcommand{\outEdges}[1]{\mathcal{E}_o(#1)}
\newcommand{\encodingFunction}[2]{h_{#1,#2}}
\newcommand{\Dimension}[1]{\text{Dim}(#1)}
\newcommand{\TransVar}[1]{\hat{#1}}
\newcommand{\RecVar}[1]{\tilde{#1}}
\newcommand{\Code}{\mathcal{C}}
\newcommand{\Directed}[1]{\vec{#1}}
\newcommand{\alphabet}{\mathcal{A}}
\newcommand{\sources}{S}
\newcommand{\decodeAlphabet}{\mathcal{B}}
\newcommand{\delay}{\delta}
\newcommand{\integer}{\mathbb{Z}}
\newcommand{\remove}[1]{#1}
\newcommand{\removeTrue}[1]{}
\newcommand{\sourceSymbol}{\alpha}
\newcommand{\sourceVec}[1]{\sourceSymbol\!\left(#1\right)}
\newcommand{\sumVec}{p}
\newcommand{\edgeVar}[1]{z_{#1}}
\newcommand{\edgeSet}{C}
\newcommand{\cut}[1]{{min-cut{$\left(#1\right)$}}}
\newcommand{\cuts}[1]{\Lambda({#1})}
\newcommand{\minCut}[1]{{min-cut{$\left(#1\right)$}}}
\newcommand{\maxRange}[1]{\hat{R}_{#1}}
\newcommand{\maxRangeB}[1]{R_{#1}}
\newcommand{\receiver}{\rho}
\newcommand{\source}{\sigma}
\newcommand{\decodFunct}{\psi}
\newcommand{\setOfMessages}{x}
\newcommand{\vecInst}{\mathbf{\alpha}}
\newcommand{\messageVecInst}{\mathbf{w}}
\newcommand{\encodeMatrix}[1]{M_{#1}}
\newcommand{\VecComp}[2]{#1_{#2}}
\newcommand{\Tree}{\mathcal{T}}
\newcommand{\tree}{T}
\newcommand{\treeIndex}[1]{J_{#1}}
\newcommand{\gap}[1]{\mathcal{C}_{\mbox{{\scriptsize gap}}}\!\left(#1\right)}
\newcommand{\codCap}[1]{\mathcal{C}_{\mbox{{\scriptsize cod}}}\!\left(#1\right)}
\newcommand{\linCodCap}[1]{\mathcal{C}_{\mbox{{\scriptsize lin}}}\!\left(#1\right)}
\newcommand{\routCap}[1]{\mathcal{C}_{\mbox{{\scriptsize rout}}}\!\left(#1\right)}
\newcommand{\edgeFunct}[1]{g_{#1}}

\newcommand{\NbinaryBlocks}[1]{{#1}^{(M)}}
\newcommand{\sumset}[1]{Q\!\left(#1\right)}
\newcommand{\zeroAt}[1]{h^{(#1)}}
\newcommand{\invZeroAt}[1]{h_{#1}^{-1}}
\newcommand{\maxWeight}[1]{\card{#1}_{H}}
\newcommand{\cardConst}[1]{\gamma(#1)}
\newcommand{\hammingWeight}[1]{\card{#1}_{H}}
\newcommand{\cardsources}{s}
\newcommand{\NOPROCESS}[1]{}
\newcommand{\FunIdentity}{f_{id}}
\newcommand{\FunMajority}{f_{maj}}
\newcommand{\FunSum}{f_{sum}}
\newcommand{\FunParity}{f_{parity}}
\newcommand{\FunMaximum}{f_{max}}
\newcommand{\FunMinimum}{f_{min}}
\newcommand{\footprintsize}{footprint size }
\newcommand{\footprintsizes}{footprint sizes }
\newcommand{\IndicesToIndex}{h}
\newcommand{\path}{P}
\newcommand{\steinerNumber}[1]{\Pi\!\left(#1\right)}
\newcommand{\lb}[1]{l\!\left(#1\right)}
\newcommand{\entropy}{\mathcal{H}}
\newcommand{\field}[1]{\mathbb{F}_{#1}}
\newcommand{\change}[1]{{\color{blue} #1}}
\newcommand{\defn}[1]{\textit{ #1}}
\begin{titlepage}

\setcounter{page}{0}

\title{Network Coding for Computing: Cut-Set Bounds%
\thanks{This work was supported by the National Science Foundation 
        and the UCSD Center for Wireless Communications.\newline
\indent The authors are with the Department of Electrical and Computer Engineering, 
        University of California, San Diego, La Jolla, CA 92093-0407. \ \ 
 (rathnam@ucsd.edu,\ massimo@ece.ucsd.edu,\ nikhil@ucsd.edu,\ zeger@ucsd.edu)
}}

\author{Rathinakumar Appuswamy \and Massimo Franceschetti \and Nikhil Karamchandani \and Kenneth Zeger}

\date{
\bigskip\bigskip
\textit{IEEE Transactions on Information Theory\\
Submitted: \submitteddate\\
Revised:  \reviseddate \\
}}

\maketitle
\begin{abstract}
The following \textit{network computing} problem is considered.
Source nodes in a directed acyclic network generate independent messages and a single receiver node
computes a target function $f$ of the messages.
The objective is to
maximize the average number of times $f$ can be computed per network usage, i.e., 
the ``computing capacity''.
The \textit{network coding} problem for a single-receiver network
is a special case of the network computing problem
in which all of the source messages must be reproduced at the receiver.
For network coding with a single receiver,
routing is known to achieve the capacity by achieving the network \textit{min-cut} upper bound.
We extend the definition of min-cut to the network computing problem
and show that the min-cut is still an upper bound on the maximum achievable rate and is tight  for computing 
(using coding)
any target function in multi-edge tree networks and for computing linear target functions in any network. 
We also 
study the bound's tightness for different classes of target functions.
\NOPROCESS{
As examples,
we demonstrate the existence of networks with irrational capacities
and networks which require non-linear coding to achieve their capacity.}
In particular, 
we give a lower bound on the computing capacity 
in terms of the Steiner tree packing number
and a differnet bound for symmetric functions.
We also show that for certain networks and target functions, the computing
capacity can be  less than an arbitrarily
small fraction of the min-cut bound.
\end{abstract}

\thispagestyle{empty}
\end{titlepage}

\clearpage

\section{Introduction}
We consider networks where source nodes 
generate independent messages and a single receiver node computes a 
target function $f$ of these messages.
The objective is to characterize the
maximum rate of computation, 
that is the maximum number of times
$f$ can be computed per network usage.
 
Giridhar and Kumar~\cite{Giridhar05} have recently stated: 
\begin{quote}
``In its most general form, computing a function in a network involves
communicating possibly correlated messages, to a specific destination,
at a desired fidelity with respect to a joint distortion criterion
dependent on the given function of interest.
This combines the complexity of
source coding of correlated sources, with rate distortion, different
possible network collaborative strategies for computing and
communication, and the inapplicability of the separation theorem
demarcating source and channel coding." 
\end{quote}
The overwhelming complexity of network computing suggests that simplifications 
be examined in order to obtain some understanding of the field.

We present a natural model of network
computing that is closely related to the network coding model of
Ahlswede, Cai, Li, and Yeung~\cite{Ahlswede-Cai-Li-Yeung-IT-Jul00,Yeung-book}.
Network coding is a
widely studied communication mechanism in the context of network
information theory.
In network coding, some nodes in the network are
labeled as sources and some as receivers.
Each receiver needs to
reproduce a subset of the messages generated by the source nodes,
and all nodes can act as relays and encode the information they
receive on in-edges, together with the information they generate if they are
sources, into codewords which are sent on their out-edges.
In
existing computer networks, the encoding operations are purely routing:
at each node, the codeword sent over an out-edge consists of
a symbol either received by the node, or generated by it if is a source.
It is known that
allowing more complex encoding than routing can in general
be advantageous in terms of communication
rate~\cite{Ahlswede-Cai-Li-Yeung-IT-Jul00,Harvey06, Ngai04}.
Network coding with a
single receiver is equivalent to a special case of our function
computing problem, namely when the function to be computed is the
identity, that is when the receiver wants to reproduce all the
messages generated by the sources.
In this paper, we study
network computation for target functions  different than the identity.

Some other approaches to network computation have also appeared in the
literature. 
In \cite{KornerMarton79,Alon-01,Vishal-Devavrat-06,Vishal-Devavrat-07,Ma08,Cuff09} 
network computing was considered as an extension of distributed source coding, 
allowing the sources to have a joint distribution and requiring that a function 
be computed with small error probability. 
A rate-distortion approach to the problem has been studied in \cite{Yamamoto82,Feng04,Vishal-07}. 
However, the complexity of network computing has restricted prior work to the analysis of elementary networks. 
Networks with noisy links were studied in
\cite{ElGamal87,Gallager88,ying07,goyal08,dutt08,Karam09,Ayaso07,Nazer-07,Ma09}
and distributed computation in networks using gossip algorithms was studied in
\cite{Kempe03, Boyd06, Aayoma08, Ayaso08, Dimakis06, Benezit07}. 
   
In the present paper,
our approach is somewhat (tangentially) related to the field of communication complexity 
\cite{Nisan-book, Yao79} 
which studies the minimum number of messages that two nodes need to exchange in 
order to compute a function of their inputs with zero error. 
Other studies of computing in networks have been considered in 
\cite{Giridhar05,Sundar07}, 
but these were restricted to the wireless communication 
protocol model of Gupta and Kumar~\cite{Gupta00}. 

In contrast, our approach is more closely associated with wired networks with independent noiseless links.
Our work is closest in spirit to the recent work of 
\cite{Ramamoorthy08,Langberg09,Rai10,Rai09} on computing the sum
(over a finite field) of source messages in networks. 
We note that in independent work,
Kowshik and Kumar\cite{Kowshik09} obtain the asymptotic maximum rate of computation 
in tree networks and present bounds for computation in networks where all nodes are sources.
   
Our main contributions are summarized in
Section~\ref{Sec:Contributions}, 
after formally introducing the network model.
\NOPROCESS{
A cut-set upper bound on the maximum rate of computing any target function in an arbitrary single-receiver network is
given in Section~\ref{Sec:UpperBound}.
In Section~\ref{Sec:LowerBounds}, 
this bound is proven to be achievable for
computing arbitrary target functions in tree networks and for computing linear target functions in arbitrary networks. 
We also provide lower bounds on the computing capacity for various target functions.
The tightness of these bounds is explored for different classes of target functions and
networks in Section~\ref{Sec:Tightness}. 
We show that there are target functions for which a computing rate of at least a constant
fraction of the min-cut can be achieved on any network, 
while for other target functions there exist networks such that a computing rate which is only an
arbitrarily small fraction of the min-cut can be achieved.
We discuss an example network and target function in detail in Section~\ref{Sec:Example}.
Finally, 
conclusions  are given in Section~\ref{Sec:conclusions}
and various lemmas are proven in an Appendix.
}
\subsection{Network model and definitions} \label{Sec:NetworkModel}
In this paper, a \textit{network} $\Network$ consists of a finite,
directed acyclic multigraph $G= (\nodes,\edges)$, 
a set of \textit{source nodes} $\sources = \{\source_1, \dots, \source_\cardsources \} \subseteq \ \nodes$, 
and a \textit{receiver} $\receiver \in \nodes$. 
Such a network is denoted by $ \Network = (G,\sources,\receiver)$.
We will assume that $\receiver \not \in \sources$
and that the graph%
\footnote{Throughout the paper, we will use ``graph" to mean a directed acyclic multigraph, 
and ``network" to mean a single-receiver network. 
We may sometimes write $\edges(G)$ to denote the edges of graph $G$.} 
$G$ contains a directed path from every node in $\nodes$ to the receiver $\receiver$. 
For each node $u \in \nodes$, 
let $\inEdges{u}$ and $\outEdges{u}$ denote the set of in-edges and out-edges of $u$ respectively. 
We will also assume
(without loss of generality)
that if a network node has no in-edges,
then it is a source node.

An \textit{alphabet} $\alphabet$ is a finite set of size at least two. 
For any positive integer $m$,
any vector $x \in \alphabet^{m}$, 
and any $i \in \{1,2,\ldots,m\}$,
let $\VecComp{x}{i}$ denote the $i$-th component of $x$.
For any index set $I = \{i_1,i_2,\ldots,i_q\} \subseteq \{1,2,\ldots,m\}$ 
with $i_1 < i_2 < \ldots < i_q$, let $\VecComp{x}{I}$ denote the vector 
$(\VecComp{x}{i_1},\VecComp{x}{i_2},\ldots,\VecComp{x}{i_q}) \in \alphabet^{\card{I}}$. 

The \textit{network computing} problem consists of a network $\Network$ and a \textit{target function} $f$ of the form
$$
f : \alphabet^{\cardsources} \longrightarrow \decodeAlphabet 
$$
(see Definition~\ref{Defn:ExTargetFns} for some examples). 
We will also assume that any target function depends on all network sources
(i.e. they cannot be constant functions of any one of their arguments).
Let $k$ and $n$ be positive integers. 
Given a network $\Network$ with source set $\sources$ and alphabet $\alphabet$, 
a \textit{message generator} is any mapping 
$$
\sourceSymbol \ : \ \sources \longrightarrow \alphabet^k .
$$
For each source $\source_i$, 
$\sourceVec{\source_i}$ is called a \textit{message vector}
and its components  
$\sourceVec{\source_i}_1, \dots, \sourceVec{\source_i}_k$
are called \textit{messages}.%
\footnote{
For simplicity, we assume that each source has exactly one message vector associated with it,
but all of the results in this paper can readily be extended to the more general case.}
\begin{definition}
A $(k,n)$ \textit{network code} for computing a target function $f$ in a network $\Network$ consists of the following: 
\begin{itemize}
\item[(i)] 
For any node $\node \in \nodes - \receiver$ and any out-edge $\edge \in \outEdges{\node}$, 
an {\it encoding function}:
\begin{align*}
h^{(e)} : \begin{cases}
&\displaystyle \left(\prod_{\hat{\edge} \in \inEdges{\node}} \alphabet^{n} \right) \times 
\alphabet^{k} \longrightarrow \alphabet^{n}  \quad \mbox{if $\node$ is a source node}\\
&\displaystyle \prod_{\hat{\edge} \in \inEdges{\node}} \alphabet^{n} \longrightarrow \alphabet^{n}  
\hspace{.9in} \mbox{otherwise}
\end{cases}
\end{align*}
\item[(ii)] 
A {\it decoding function}:
$$
\decodFunct
: \ \prod_{j=1}^{\card{\inEdges{\receiver}}} \alphabet^{n} \longrightarrow \decodeAlphabet^{k} .
$$
\end{itemize}
\end{definition}
Given a $(k,n)$ network code,
every edge $\edge \in \edges$ \textit{carries a vector} $z_{\edge}$ of at most $n$ alphabet symbols,%
\footnote{By default, we will  assume that edges carry exactly $n$ symbols.}
which is obtained by evaluating the encoding function 
$h^{(e)}$ on the set of vectors carried by the in-edges to the node and the node's message vector if it is a source. 
The objective of the
receiver is to compute the target function $f$ of the source messages, 
for any arbitrary message generator $\sourceSymbol$.
More precisely, the receiver constructs a vector of $k$ alphabet symbols such that for each $i \in \{1, 2,\ldots, k\}$,
the $i$-th component of the receiver's computed vector equals the value of the desired 
target function $f$ applied to the $i$-th components of the source message vectors,
for any choice of message generator $\sourceSymbol$.
Let $e_1, e_2, \ldots, e_{\card{\inEdges{\receiver}}}$ denote the in-edges of the receiver. 
\begin{definition}
A $(k,n)$ network code is called  \textit{a solution for computing $f$ in $\Network$} (or simply {\it a $(k, n)$ solution})
if the decoding function $\decodFunct$ is such that
for each $j \in \{1, 2,\ldots, k\}$ and for every message generator $\sourceSymbol$,
we have
\begin{align}
\decodFunct\left(\edgeVar{e_1},\cdots,\edgeVar{e_{\card{\inEdges{\receiver}}}}\right)_j 
&= f\left(\sourceVec{\source_1}_j,\cdots,\sourceVec{\source_{\cardsources}}_j\right). \label{Eq:decodingFunction}
\end{align}
If there exists a $(k, n)$ solution, 
we say the rational number $k/n$ is an {\it achievable computing rate}.
\end{definition}
In the network coding literature, 
one definition of the \textit{coding capacity} of a network is the supremum of all achievable coding rates 
\cite{Cannons-Dougherty-Freiling-Zeger05, Dougherty-Freiling-Zeger06}. 
We adopt an analogous definition for computing capacity. 
\begin{definition}
The \textit{computing capacity} of a network $\Network$ with respect to target function $f$ is
$$
\codCap{\Network,f} \; = \; \sup  \Big\{ \frac{k}{n} \ : \ 
    \mbox{$\exists$  $(k,n)$ network code for computing $f$ in $\Network$}\Big\}.
$$
\label{def:cap}
\end{definition}
Thus, the computing capacity is the supremum of all achievable computing rates 
for a given network $\Network$ and a target function $f$. 
Some example target functions are defined below.
\begin{definition} \label{Defn:ExTargetFns}\ \\
\begin{table}[hht] 
\begin{center}
\renewcommand{\arraystretch}{1.2} 
\begin{tabular}{|c||c|c|c|}
\hline 
Target function  $f$			& 			Alphabet $\mathcal{A}$					&			 $f\left(x_{1}, \ldots , x_{s} \right)$	& 		Comments 	 \\
\hline
\hline 
\defn{identity}				& 		arbitrary   				& $\left(x_{1},  \ldots , x_{s} \right)$	& 	 \\
\hline 
\defn{arithmetic sum}	& 	$\{0,1,\ldots,q-1\}$	& $x_1 + x_2 + \cdots + x_s$	& `$+$' is ordinary integer addition	  \\
\hline
\defn{mod $r$ sum}     &   $\{0,1,\ldots,q-1\}$	& $x_{1} \oplus x_{2} \oplus \ldots \oplus x_{s}$	& $\oplus$ is $\bmod$ $r$ addition	 \\
\hline 
\defn{histogram}				& 	$\{0,1,\ldots,q-1\}$	& $\left(c_0,c_1,\ldots,c_{q - 1}\right)$	& $c_i = \left|\left\{j : x_j = i\right\}\right|$ for each $i \in \mathcal{A}$	  \\
\hline
\defn{linear}					&  any finite field 								& $a_{1} x_{1} + a_{2} x_2 + \ldots + a_{s} x_{s}$	& 	arithmetic performed in the field \\
\hline 
\defn{maximum}					&  any ordered set 				& $\max \left\{x_{1}, \ldots , x_{s} \right\}$	& 	 \\
\hline
\end{tabular}
\end{center}
\label{Tab:checkFunction}
\end{table} 
\end{definition}
\begin{definition}
\label{Defn:EquivRelation}
For any target function 
$f  :  \alphabet^{\cardsources} \longrightarrow \decodeAlphabet$, 
any index set 
$I \subseteq \{1,2,\ldots,\cardsources\}$, 
and 
any $a, b \in \alphabet^{\card{I}}$, 
we write $a \equiv b$ if for every 
$x, y \in \alphabet^{\cardsources}$, 
we have $f(x) = f(y)$
whenever
$\VecComp{x}{I} = a$, $\VecComp{y}{I} = b$, 
and
$\VecComp{x}{j} = \VecComp{y}{j}$ for all $j\not\in I$. 
\end{definition}
It can be verified that $\equiv$ is an equivalence relation%
\footnote{ 
Witsenhausen \cite{Witsenhausen76} represented this equivalence relation 
in terms of the independent sets of a characteristic graph
and his representation has been used in various problems related to function computation 
\cite{Vishal-Devavrat-06, Vishal-Devavrat-07, Alon-01}. 
Although $\equiv$ is defined with respect to a particular index set $I$ and a function $f$, 
we do not make this dependence explicit -- the values of $I$ and $f$ will be clear from the context. } 
for every $f$ and $I$. 
\begin{definition}
\label{Defn:NoOfEquivClasses}
For every $f$ and $I$,
let $\maxRangeB{I,f}$ denote the total number of equivalence classes induced by $\equiv$
and let
$$
\Phi_{I, f} : \alphabet^{\card{I}} \longrightarrow \left\{1,2,\ldots,\maxRangeB{I,f}\right\}
$$
be any function such that $\Phi_{I, f}(a) = \Phi_{I, f}(b)$ iff $a \equiv b$. 
\end{definition}
That is, 
$\Phi_{I, f}$ assigns a unique index to each equivalence class, and
$$
\maxRangeB{I, f} = \card{ \left\{ \Phi_{I,f}(a) : a \in \alphabet^{\card{I}} \right\} } .
$$
The value of $\maxRangeB{I, f}$ is independent of the choice of $\Phi_{I,f}$.
We call $\maxRangeB{I, f}$ the \textit{\footprintsize} of $f$ with respect to $I$. 
\remove{
\begin{figure}[ht]
\begin{center}
\psfrag{X}{\Large$X$}
\psfrag{Y}{\Large$Y$}
\psfrag{x}{\large$x$}
\psfrag{y}{\large$y$}
\psfrag{g}{\large$g(x)$}
\psfrag{f}{\large$f(x, y)$}
\scalebox{.7}{\includegraphics{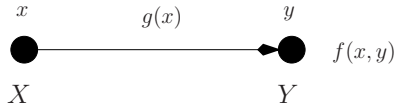}}
\end{center}
\caption{$X$, $Y$ are two sources with messages $x$ and $y$ respectively. 
$X$ communicates $g(x)$ to $Y$ so that $Y$ can compute a function $f$ of 
$x$ and $y$.}
\label{Fig:EquivClass}
\end{figure}
}

\begin{remark} 
Let $I^{c} = \{1,2,\ldots,\cardsources\} - I$. 
The \footprintsize $\maxRangeB{I, f}$ has the following interpretation
(see Figure~\ref{Fig:EquivClass}). 
Suppose a network has two nodes, $X$ and $Y$,
and both are sources.
A single directed edge connects $X$ to $Y$.
Let $X$ generate $x \in \alphabet^{\card{I}}$ and 
$Y$ generate $y \in \alphabet^{\card{I^{c}}}$. 
$X$ communicates a function $g(x)$ of its input, 
to $Y$ so that $Y$ can compute 
$f(a)$ where $a \in \alphabet^{\cardsources}$, $\VecComp{a}{I} = x$, and $\VecComp{a}{I^{c}} = y$. 
Then for any $x, \hat{x} \in \alphabet^{\card{I}}$ such that $x \not\equiv \hat{x}$, 
we need $g(x) \ne g(\hat{x})$. 
Thus $\card{ g\left( \alphabet^{\card{I}}\right) } \ge \maxRangeB{I, f}$, 
which implies a lower bound on a certain amount of ``information" that $X$ needs to send to $Y$ to 
ensure that it can compute the function $f$. 
Note that $g = \Phi_{I, f}$ achieves the lower bound.   
We will use this intuition to establish a cut-based upper bound on the computing capacity 
$\codCap{\Network, f}$ of any network $\Network$ with respect to any target function 
$f$,
and to devise a capacity-achieving scheme for computing any target function in multi-edge tree networks.
\end{remark}    

\begin{definition}
A set of edges $C \subseteq \edges$ in network $\Network$
 is said to \textit{separate} 
sources $\source_{m_1}, \ldots, \source_{m_d}$
from the receiver $\receiver$, 
if for each $i \in \{1, 2,\ldots, d\}$,
every directed path from
$\source_{m_i}$ to $\receiver$ contains at least one edge in $C$.
The set $C$ is said to be a \textit{cut} 
in $\Network$ if it separates at least one
source from the receiver.
For any network $\Network$,
define $\cuts{\Network}$ to be the collection of all cuts in $\Network$.
For any cut $C \in \cuts{\Network}$ and any target function $f$, 
define
\begin{align} 
I_C &= \left\{i : \mbox{$C$ separates $\source_i$ from the receiver} \right\}\notag\\
\maxRangeB{C,f} &= \maxRangeB{I_C, f} \label{Eq:maxRangeDef}.
\end{align}
\end{definition}
Since target functions depend on all sources,
we have $\maxRangeB{C,f} \ge 2$ for any cut $C$ and any target function $f$.
The \footprintsizes $\maxRangeB{C,f}$ for some example target functions are computed below. 
A \textit{multi-edge tree} is a graph such that for every node 
$\node \in \nodes$, 
there exists a node $u$ such that all the out-edges of $\node$ are in-edges to $u$, i.e., 
$\outEdges{\node} \subseteq \inEdges{u}$
(e.g. see Figure~\ref{Fig:multiEdgeTreeExample}).

\begin{figure}[ht]
\begin{center}
\scalebox{.6}{\includegraphics{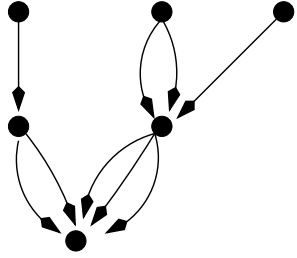}}
\end{center}
\caption{An example of a multi-edge tree.} \label{Fig:multiEdgeTreeExample}
\end{figure}
\subsection{Classes of target functions}
We study the following four classes of target functions: 
(1) divisible,
(2) symmetric,
(3) $\lambda$-exponential,
(4) $\lambda$-bounded.
\begin{definition}\label{Defn:DivisibleFunc}
A target function $f : \alphabet^{\cardsources} \longrightarrow \decodeAlphabet$
is \textit{divisible} if 
for every index set $I \subseteq \{1, \dots, \cardsources\}$,
there exists a
finite set $\decodeAlphabet_{I}$
and a function
$f^{I} : \alphabet^{\card{I}} \longrightarrow \decodeAlphabet_{I}$
such that the following hold:
\begin{itemize}
\item[(1)]
$f^{\{1, \dots, \cardsources\}} = f$ 
\item[(2)]
$\card{ f^{I}\left(\alphabet^{\card{I}}\right) }  \le  \card{ f\left(\alphabet^{\cardsources}\right) }$ 
\item[(3)]
For every partition $\{I_1, \ldots, I_{\gamma}\}$ of $I$,
there exists a function \\
$g: \decodeAlphabet_{I_1} \times \dots \times \decodeAlphabet_{I_{\gamma}} \longrightarrow \decodeAlphabet_I$ 
such that for every
$\setOfMessages \in \alphabet^{\card{I}}$, we have\\
$f^{I}\!\left(\setOfMessages\right)
= g \left(f^{I_1}\!\left(\VecComp{\setOfMessages}{I_1}\right),
  \ldots,
  f^{I_{\gamma}}\!\left(\VecComp{\setOfMessages}{I_{\gamma}}\right)\right)$.
\end{itemize}
\end{definition}
Examples of divisible target functions include the
identity, maximum, $\bmod$ $r$ sum, and arithmetic sum. 

Divisible functions have been studied previously%
\footnote{The definitions 
in \cite{Giridhar05, Sundar07}
are similar to ours but slightly more restrictive.} 
by
Giridhar and Kumar \cite{Giridhar05} 
and
Subramanian, Gupta, and Shakkottai \cite{Sundar07}.
Divisible target functions can be computed in networks in a divide-and-conquer fashion as follows.
For any arbitrary partition
$\{I_1,\ldots,I_{\gamma}\}$ of the source indices $\{1, \dots, \cardsources\}$,
the receiver $\receiver$ can evaluate the target function $f$ by
combining evaluations of
$f^{I_1}, \ldots, f^{I_{\gamma}}$.
Furthermore, 
for every $i=1, \ldots, \gamma$,
the target function $f^{I_i}$ can be evaluated
similarly by partitioning $I_i$ and this process can be repeated until the function value is obtained.

\begin{definition}\label{Defn:SymmetricFunc}
A target function $f : \alphabet^{\cardsources} \longrightarrow \decodeAlphabet$
is \textit{symmetric} if for any permutation
$\pi$ of $\{1,2,\ldots,\cardsources\}$ and
any vector $x \in \alphabet^{\cardsources}$,  
$$
f(\VecComp{x}{1},\VecComp{x}{2},\ldots,\VecComp{x}{\cardsources}) = 
f(\VecComp{x}{\pi(1)},\VecComp{x}{\pi(2)},\ldots,\VecComp{x}{\pi(\cardsources)}).
$$
\end{definition}
That is, 
the value of a symmetric target function is invariant with respect to the order of its arguments and hence, 
it suffices to evaluate the histogram target function for computing any symmetric target function.
Examples of symmetric functions include the arithmetic sum, maximum, and $\bmod$ $r$ sum.
Symmetric functions have been studied in the context of computing in networks by
Giridhar and Kumar \cite{Giridhar05},
Subramanian, Gupta, and Shakkottai \cite{Sundar07},
Ying, Srikant, and Dullerud \cite{ying07},
and \cite{Karam09}. 
\NOPROCESS{
For any target function $f : \alphabet^{\cardsources} \longrightarrow \decodeAlphabet$
and any index set $I \subseteq \{1, 2,\ldots, \cardsources\}$,
it can be verified that the \footprintsize satisfies
\begin{align}
\maxRangeB{I, f} \le \card{\alphabet}^{\card{I}} .
\label{eq:2}
\end{align}
}
\begin{definition}
\label{Defn:LambdaExpFunc}
Let $\lambda \in (0, 1]$.
A target function $f : \alphabet^{\cardsources} \longrightarrow \decodeAlphabet$
is said to be
\textit{$\lambda$-exponential} if its \footprintsize satisfies
$$
\maxRangeB{I, f} \ge \card{\alphabet}^{\lambda \card{I}} \ \mbox{for every } I \subseteq \{1, 2,\ldots, \cardsources\} .
$$
%
Let $\lambda \in (0,\infty)$.
A target function $f  :  \alphabet^{\cardsources} \longrightarrow \decodeAlphabet$
is said to be
\textit{$\lambda$-bounded} if its \footprintsize satisfies
$$
\maxRangeB{I, f} \le \card{\alphabet}^{\lambda} \ \mbox{for every } I \subseteq \{1, 2,\ldots, \cardsources\} .
$$
\end{definition}
\begin{example}
The following facts are easy to verify:
\begin{itemize}
\item
The identity function is $1$-exponential.

\item
Let $\alphabet$ be an ordered set.
The maximum (or minimum) function is $1$-bounded.

\item
Let $\alphabet = \{0,1,\ldots,q-1\}$ where  $q \ge 2$. 
The $\bmod$ $r$ sum target function with $q \ge r \ge 2$ is $\log_{q} r$-bounded.

\end{itemize}
\end{example}
\begin{remark}
Giridhar and Kumar \cite{Giridhar05} defined two classes of functions: 
\emph{type-threshold} and \emph{type-sensitive} functions. 
Both are sub-classes of symmetric functions. 
In addition, type-threshold functions are also divisible and $c$-bounded, 
for some constant $c$ that is independent of the network size. 
However, \cite{Giridhar05} uses a model of interference for simultaneous transmissions 
and their results do not directly compare with ours. 
\end{remark}
%
Following the notation in Leighton and Rao \cite{LeightonRao99},
the \textit{min-cut} of any network $\Network$ with unit-capacity edges is
\begin{align}\label{Eq:min-cut}
\mbox{\minCut{\Network}} = \underset{C \in \cuts{\Network}}{\min} \ \frac{\card{C}}{\card{I_{C}}} .
\end{align}
A more general version of the network min-cut plays a fundamental role in the field of multi-commodity flow 
\cite{LeightonRao99, Vazirani04}. 
The min-cut provides an upper bound on the maximum flow for any multi-commodity flow problem. 
The min-cut is also referred to as ``sparsity" by some authors, 
such as
Harvey, Kleinberg, and Lehman \cite{Harvey06}
and 
Vazirani \cite{Vazirani04}.
We next generalize the definition in \eqref{Eq:min-cut} to the network computing problem.
\begin{definition}
If $\Network$ is a network and $f$ is a target function, then define
\begin{align}
\label{Eq:min-f-cut}
\mbox{\cut{\Network,f}} = \underset{C \in \cuts{\Network}}{\min} 
 \ \frac{\card{C}}{\log_{\card{\alphabet}}\maxRangeB{C,f}} .
\end{align}
\end{definition}
\begin{example}
\
\begin{itemize}
\item If $f$ is the identity target function, 
then 
$$
\mbox{\cut{\Network,f}} = \underset{C \in \cuts{\Network}}{\min} 
 \ \frac{\card{C}}{\card{I_C}}.
$$
Thus for the identity function, 
the definition of min-cut in \eqref{Eq:min-cut} and \eqref{Eq:min-f-cut} coincide.
\item
Let $\alphabet = \{0,1,\ldots,q -1\}$.
If $f$ is the arithmetic sum target function, 
then 
\begin{equation}
\label{Eq:min-cutSum}
\mbox{\cut{\Network,f}} = \underset{C \in \cuts{\Network}}{\min} 
 \ \frac{\card{C}}{\log_{q} \left((q - 1)\card{I_C} + 1\right)}.
\end{equation}
\item
Let $\alphabet$ be an ordered set. 
If $f$ is the maximum target function, 
then 
$$
\mbox{\cut{\Network,f}} = \underset{C \in \cuts{\Network}}{\min} \card{C}.
$$
\end{itemize}
\end{example}

\subsection{Contributions}
\label{Sec:Contributions}

The main results of this paper are as follows. 
In Section~\ref{Sec:UpperBound},
we show
(Theorem~\ref{Th:upperBoundOnCodingCapacity})
that for any network $\Network$ and any target function $f$, 
the quantity \cut{\Network,f} is an upper bound on the computing capacity $\codCap{\Network,f}$.
In Section~\ref{Sec:LowerBounds},
we note that the computing capacity for any network with respect to the identity target function is equal 
to the min-cut upper bound 
(Theorem~\ref{Th:AchievabilityMinCutIdentity}).
We show that the min-cut bound on computing capacity can also be achieved 
for all networks with linear target functions over finite fields
(Theorem~\ref{Th:ModuloSumCodCap})
and for all multi-edge tree networks with any target function
(Theorem~\ref{Th:codCapTree}).
For any network and any target function, 
a lower bound on the computing capacity is given in terms of the Steiner tree packing number
(Theorem~\ref{Th:LowerBndGeneral}).
Another lower bound is given for networks 
with symmetric target functions
(Theorem~\ref{Th:SymmFuncCodCap}).
In Section~\ref{Sec:Tightness}, 
the tightness of the above-mentioned bounds is analyzed 
for divisible (Theorem~\ref{thm:divisible}), 
symmetric (Theorem~\ref{Th:SymmFuncGap}), 
$\lambda$-exponential (Theorem~\ref{Th:GapLambdaExp}),
and
$\lambda$-bounded (Theorem~\ref{Th:GapLambdaBded}) target functions.
For $\lambda$-exponential target functions,
the computing capacity is at least $\lambda$ times the min-cut.
If every non-receiver node in a network is a source,
then for $\lambda$-bounded target functions 
the computing capacity is at least a constant times the min-cut divided by $\lambda$.
It is also shown, with an example target function, that 
there are networks for which the computing capacity is less than an arbitrarily small fraction of the min-cut bound 
(Theorem~\ref{Th:UnboundedGap}). 
In Section~\ref{Sec:Example}, 
we discuss an example network and target function in detail to illustrate the above bounds. 
In Section~\ref{Sec:conclusions}, conclusions are given 
and various lemmas are proven in the Appendix.



\section{Min-cut upper bound on computing capacity}
\label{Sec:UpperBound}
The following shows that the maximum rate of computing a target function $f$ in a network $\Network$ 
is at most \cut{\Network,f}.
\begin{theorem} \label{Th:upperBoundOnCodingCapacity}
If $\Network$ is a network with target function $f$, then
$$
\codCap{\Network,f} \leq \mbox{\textup{\cut{\Network,f}}} .
$$
\end{theorem}
\begin{proof}
Let the network alphabet be $\alphabet$ and
consider any $(k,n)$ solution for computing $f$ in $\Network$.
%
Let $C$ be a cut and for each $i \in \{1,2,\ldots,k\}$, let
$a^{(i)}, b^{(i)} \in \alphabet^{\card{I_{C}}}$.
Suppose $j \in \{1,2,\ldots,k\}$  is such that $a^{(j)} \not\equiv b^{(j)}$,
where $\equiv$ is the equivalence relation from
Definition~\ref{Defn:EquivRelation}.
Then there exist $x, y \in \alphabet^{\cardsources}$ satsifying:
$f(x) \neq f(y)$,
$\VecComp{x}{I_C} = a^{(j)}$, 
$\VecComp{y}{I_C} = b^{(j)}$, 
and
$\VecComp{x}{i} = \VecComp{y}{i}$ for every $i \not\in I_{C}$.

The receiver $\receiver$ can compute the target function $f$ only if, 
for every such pair $\left\{a^{(1)},\ldots,a^{(k)}\right\}$ and $\left\{b^{(1)},\ldots,b^{(k)}\right\}$ 
corresponding to the message vectors generated by the sources in $I_C$,
the edges in cut $C$ carry distinct vectors. 
Since the total number of equivalence classes for the relation $\equiv$ equals the \footprintsize $\maxRangeB{C,f}$,
the edges in cut $C$ should carry at least $\left(\maxRangeB{C,f}\right)^{k}$ distinct vectors. 
Thus, 
we have 
$$
\alphabet^{n\card{C}} \ge \left(\maxRangeB{C,f}\right)^{k} 
$$
and hence for any cut $C$,
$$
\frac{k}{n} \le \frac{\card{C}}{\log_{\card{\alphabet}}\maxRangeB{C,f}} .
$$
Since the cut $C$ is arbitrary, the result follows from
Definition~\ref{def:cap} and \eqref{Eq:min-f-cut}.
\end{proof}
The min-cut upper bound has the following intuition.
Given any cut $C \in \cuts{\Network}$,
at least $\log_{\card{\alphabet}}\maxRangeB{C,f}$
units of information need to be sent across the cut
to successfully compute a target function $f$. 
In subsequent sections,
we study the tightness of this bound for different classes of functions and networks. 


\section{Lower bounds on the computing capacity}
\label{Sec:LowerBounds}
%
\NOPROCESS{
First in Theorems~\ref{Th:AchievabilityMinCutIdentity} - \ref{Th:codCapTree}, 
we obtain lower bounds on the computing capacity $\codCap{\Network, f}$ 
for three special cases, which coincide with the \textup{\cut{\Network,f}} upper bound, 
thus demonstrating that the computing capacity is equal to \textup{\cut{\Network,f}} for these cases. 
We then establish a general lower bound on the computing capacity for arbitrary networks and target functions 
(Theorem~\ref{Th:LowerBndGeneral}) and another lower bound specifically for symmetric target functions 
(Theorem~\ref{Th:SymmFuncCodCap}). 
In Section~\ref{Sec:Tightness}, 
we use Theorem~\ref{Th:LowerBndGeneral} and Theorem~\ref{Th:SymmFuncCodCap} 
to establish further results on the tightness of the \textup{\cut{\Network,f}} 
upper bound for different classes of target functions. 
}
The following result shows that the computing capacity of any network 
$\Network$ with respect to the identity target function equals the coding capacity for ordinary network coding. 
\begin{theorem}\label{Th:AchievabilityMinCutIdentity}
If $\Network$ is a network with the identity target function $f$, then
$$
\codCap{\Network,f} = \mbox{\textup{\cut{\Network,f}}} 
= \mbox{\textup{\minCut{\Network}}} .
$$
\end{theorem}
\begin{proof}
Rasala Lehman and Lehman \cite[p.6, Theorem~4.2]{AprilLehman-EricLehman-04} showed that for any single-receiver network, 
the conventional coding capacity 
(when the receiver demands the messages generated by all the sources) always equals the \minCut{\Network}.
Since the target function $f$ is the identity, the computing
capacity is the coding capacity and \minCut{\Network,f} = \minCut{\Network}, so the result follows.
\end{proof}
\begin{theorem}\label{Th:ModuloSumCodCap}
If $\Network$ is a network with a finite field alphabet and with a linear target function $f$, then
$$
\codCap{\Network, f} = \mbox{\textup{\cut{\Network,f}}}.
$$
\end{theorem}
\begin{proof}
Follows from \cite[Theorem~2]{Rai10}. 
\end{proof}
Theorems~\ref{Th:AchievabilityMinCutIdentity} and \ref{Th:ModuloSumCodCap} 
demonstrate the achievability of the min-cut bound for arbitrary networks with particular target functions. 
In contrast, the following result demonstrates the achievability of the min-cut bound for arbitrary 
target functions and a particular class of networks.
%
%
The following theorem concerns multi-edge tree networks, which were defined 
in Section~\ref{Sec:NetworkModel}.
\begin{theorem}\label{Th:codCapTree}
If $\Network$ is a multi-edge tree network with 
target function $f$, then
$$
\codCap{\Network,f} = \mbox{\textup{\cut{\Network,f}}} .
$$
\end{theorem}
\begin{proof}
Let $\alphabet$ be the network alphabet. 
From Theorem~\ref{Th:upperBoundOnCodingCapacity}, 
it suffices to show that
$\codCap{\Network,f} \ge \mbox{\cut{\Network,f}}$.
Since $\outEdges{\node}$ is a cut for node $\node \in \nodes - \receiver$, 
and using \eqref{Eq:maxRangeDef},
we have
\begin{align}
\mbox{\cut{\Network,f}} & \le \underset{\node \ \in \ \nodes - \receiver }{\min} \ 
   \frac{ \card{\outEdges{\node}} }{ \log_{\card{\alphabet}}\maxRangeB{\outEdges{\node},f} }.
%
\label{eq:4}
\end{align}
%
\NOPROCESS{
To evaluate $f(x)$ for $x \in \alphabet^{\cardsources}$, 
it suffices for receiver $\receiver$ to identify the equivalence class 
$\Phi_{I_{\inEdges{\receiver}}, f}\left(x\right)$ of $x$.
\begin{align} \label{Eq:propertyMultiedge}
I_{\node} = \bigcup_{(u, \node) \in \edges} I_{u} .
\end{align}
}
%
Consider any positive integers $k, n$ such that 
\begin{equation}
\label{Eq:ChosenRate}
 \frac{k}{n} \le \underset{\node \ \in \ \nodes - \receiver }{\min} \ 
     \frac{ \card{\outEdges{\node}} }{  \log_{\card{\alphabet}} \maxRangeB{I_{\outEdges{\node}}, f} } \ .
\end{equation}
Then we have 
\begin{equation}
\label{Eq:SurjCondn}
\card{\alphabet}^{\card{\outEdges{v}}n} \ge \maxRangeB{I_{\outEdges{\node}}, f}^k \ \mbox{ for every node } 
v \in \nodes - \receiver . 
\end{equation}
We outline a $(k,n)$ solution for computing $f$ in the multi-edge tree 
network $\Network$. 
Each source $\source_i \in \sources$ generates a message vector $\sourceVec{\source_i} \in \alphabet^{k}$.
Denote the vector of $i$-th components of the source messages by 
$$
x^{(i)} = \left(\VecComp{\sourceVec{\source_{1}}}{i},\cdots,\VecComp{\sourceVec{\source_{\cardsources}}}{i}\right).
$$
Every node $\node \in \nodes - \{\receiver\}$ sends out a unique index 
(as guaranteed by \eqref{Eq:SurjCondn})
over  $A^{\card{\outEdges{\node}}n}$ corresponding to the set of equivalence classes 
\begin{align} \label{Eq:Classes}
\Phi_{I_{\outEdges{\node}},f} (x^{(l)}_{I_{\outEdges{\node}}})  \ \mbox{ for }\ l\in\{1,\cdots, k\}.
\end{align}

If $\node$ has no in-edges, then by assumption, it is a source node, say $\source_j$.
The set of equivalence classes in \eqref{Eq:Classes} is a function of its own messages 
$\VecComp{\sourceVec{\source_{j}}}{l}$ for $l\in\{1,\ldots,k\}$.
On the other hand if $\node$ has in-edges,  
then let $u_1, u_2, \cdots, u_j$ be the nodes with out-edges to $\node$.
For each $i \in \{1,2,\cdots,j\}$, using the uniqueness of the index received from $u_i$,
node $\node$ recovers the equivalence classes
\begin{align} \label{Eq:incomingClasses}
\Phi_{I_{\outEdges{u_i}},f} (x^{(l)}_{I_{\outEdges{u_i}}})   \ \mbox{ for }\ l\in\{1,\cdots, k\}.
\end{align}
Furthermore, the equivalence classes in \eqref{Eq:Classes}
can be identified by $\node$ from  
the equivalance classes in \eqref{Eq:incomingClasses} 
(and $\sourceVec{\node}$ if $\node$ is a source node) using the fact that for a multi-edge tree network $\Network$,
we have a disjoint union
\begin{align*} 
I_{\outEdges{\node}} = \bigcup_{i=1}^{j} I_{\outEdges{u_i}}. 
\end{align*}
%

If each node $\node$ follows the above steps, then the receiver $\receiver$ can identify the equivalence classes 
$\Phi_{I_{\inEdges{\receiver}}, f}\left(x^{(i)}\right)$ for $i\in\{1,\ldots,k\}$. 
The receiver can evaluate $f(x^{(l)})$ for each $l$ from these equivalence classes.
The above solution achieves a computing rate of  $k/n$.
From \eqref{Eq:ChosenRate}, it follows that 
\begin{equation}\label{eq:CodCapTree}
\codCap{\Network,f} \ge \underset{\node \ \in \ \nodes - \receiver }{\min} \ 
   \frac{ \card{\outEdges{\node}} }{  \log_{\card{\alphabet}} \maxRangeB{I_{\outEdges{\node}},f} } .
\end{equation}
\end{proof}
%
We next establish a general lower bound on the computing capacity for arbitrary target functions 
(Theorem~\ref{Th:LowerBndGeneral}) and then another lower bound specifically for symmetric target functions 
(Theorem~\ref{Th:SymmFuncCodCap}). 

For any network 
$\Network = \left(G, \sources, \receiver\right)$ with $G = (\nodes,\edges)$, 
define a 
\textit{Steiner tree}%
\footnote{Steiner trees are well known in the literature for undirected graphs. 
For directed graphs a ``Steiner tree problem''
has been studied and our definition is consistent with 
such work (e.g., see \cite{Jain-03}).} 
of $\Network$ to be a minimal (with respect to nodes and edges) subgraph of $G$ containing $\sources$
and $\receiver$ such that every source in $\sources$ has a directed path to the receiver 
$\receiver$. 
Note that every non-receiver node in a Steiner tree has exactly one out-edge. 
Let $\Tree(\Network)$ denote the collection of all Steiner trees in $\Network$. 
For each edge $e \in \edges(G)$, 
let $\treeIndex{e} = \{i \, : \, t_i \in  \Tree(\Network) \mbox{ and } e \in \edges(t_i)\}$.
The \textit{fractional Steiner tree packing number} $\steinerNumber{\Network}$ 
is defined as the linear program
\begin{equation}
\label{Eq:FracSteiner}
\steinerNumber{\Network} = \max \sum_{t_i \in \Tree(\Network)} u_{i} \quad \mbox{subject to} \quad  
\begin{cases}
& \!\!\!\displaystyle u_i \ge 0 \ \ \ \forall \ t_i \in \Tree(\Network) \ , \\ 
& \!\!\!\displaystyle \sum_{i \in \treeIndex{e}} u_{i} \le 1 \ \ \ \forall \ e \in \edges(G).
\end{cases}
\end{equation}
Note that $\steinerNumber{\Network} \ge 1$ for any network $\Network$, 
and the maximum value of the sum in \eqref{Eq:FracSteiner} is attained at 
one or more vertices of the closed polytope 
corresponding to the linear constraints. 
Since all coefficients in the constraints are rational, 
the maximum value in \eqref{Eq:FracSteiner} can be attained with 
rational $u_i$'s. 
The following theorem provides a lower bound%
\footnote{
In order to compute the lower bound, 
the fractional Steiner tree packing number 
$\steinerNumber{\Network}$ can be evaluated using linear programming. 
Also note that if we construct the \textit{reverse multicast network} 
by letting each source in the original network $\Network$ become a receiver, 
letting the receiver in the $\Network$ become the only source, 
and reversing the direction of each edge, 
then it can be verified that the routing capacity 
for the reverse multicast network is equal to $\steinerNumber{\Network}$.}
on the computing capacity for any network $\Network$ with respect to a target function $f$ and uses the quantity $\steinerNumber{\Network}$.  
In the context of computing functions, $u_i$  in the above linear program indicates the fraction of the time 
the edges in tree $t_i$ are used to compute the desired function. 
The fact that every edge in the network has unit capacity implies  
$\sum_{i \in \treeIndex{e}} u_{i} \le 1$.
%
%
\begin{lemma}\label{Lemma:equivalenceOfCuts}
For any Steiner tree $G'$ of a network $\Network$, 
let $\Network'= (G', \sources, \receiver)$.
Let $C'$ be a cut in $\Network'$.
Then there exists a cut $C$ in $\Network$ such 
that $I_{C} = I_{C'}$.
\end{lemma}
(Note that $I_{C'}$ is the set indices of sources separated 
in $\Network'$ by $C'$. The set $I_{C'}$  may differ from
the indices of sources separated in $\Network$ by $C'$.)
\begin{proof}
Define the cut  
\begin{equation}
\label{Eq:DefinedCut}
C = \bigcup_{i' \in I_{C'}} \outEdges{\source_{i'}} .
\end{equation}
$C$ is the collection of out-edges in $\Network$
of a set of sources disconnected by the cut $C'$ in $\Network'$. 
If $i \in I_{C'}$, then, by 
\eqref{Eq:DefinedCut}, $C$ disconnects $\source_i$ from $\receiver$ in $\Network$, and thus $I_{C'} \subseteq I_{C}$.

Let $ \source_i$ be a source.
 such that $i \in I_{C}$ and 
Let $\path$ be a path  from $\source_i$ to $\receiver$ in $\Network$.
From $\eqref{Eq:DefinedCut}$, it follows that there exists 
$i' \in I_{C'}$ such that 
$\path$ contains at least one edge in $\outEdges{\source_{i'}}$.
If $\path$ also lies in $\Network'$ and does not contain any edge
in $C'$, then $\source_{i'}$ has a path to $\receiver$
in $\Network'$ that does not contain any edge in $C'$, thus contradicting 
the fact that $\source_{i'} \in I_{C'}$. 
Therefore, either $\path$ does not lie in $\Network'$ or 
$\path$ contains an edge in $C'$. 
Thus $\source_i \in I_{C'}$, i.e., $I_C \subseteq I_{C'}$.
\end{proof}
\begin{theorem}\label{Th:LowerBndGeneral}
If $\Network$ is a network with alphabet $\alphabet$ and target function 
$f$, then
$$
\codCap{\Network,f} \ge \steinerNumber{\Network} 
\cdot \underset{C \in \cuts{\Network} }{\min} \ \frac{1}{\log_{\card{\alphabet}} \maxRangeB{C,f}} . 
$$
\end{theorem}
\begin{proof}
Suppose $\Network = (G, \sources, \receiver)$. 
Consider a Steiner tree $G' = \left(\nodes', \edges'\right)$ of $\Network$, 
and let $\Network'=(G',\sources,\receiver)$. 
From Lemma \ref{Lemma:equivalenceOfCuts} 
(taking $C'$ to be $\outEdges{\node}$ in $\Network'$),
we have 
\begin{equation}
\label{Eq:NodeToCutCopy}
\forall \ \node \in \nodes' - \receiver, \ \exists \ C \in \cuts{\Network} \mbox{ such that }
I'_{\outEdges{\node}} = I_{C}.
\end{equation}
Now we lower bound the computing capacity for the network 
$\Network'$ with respect to target function $f$.  
\begin{align}
\label{eq:CodCapSteinerTreeCopy}
\codCap{\Network', f} & = \mbox{\cut{\Network',f}} & \Comment{Theorem~\ref{Th:codCapTree}} \\
\nonumber
& =  \underset{\node \ \in \ \nodes' - \receiver }{\min} \ 
     \frac{ 1 }{  \log_{\card{\alphabet}} \maxRangeB{I'_{\outEdges{\node}}, f} }  
         & \Comment{Theorem~\ref{Th:upperBoundOnCodingCapacity}, \eqref{eq:4}, \eqref{eq:CodCapTree}} \\   
& \ge \underset{C \in \cuts{\Network} }{\min} \ \frac{1}{\log_{\card{\alphabet}} \maxRangeB{I_C,f}}  & \Comment{\eqref{Eq:NodeToCutCopy}}. \label{eq:CodCapTreeLBCopy}
\end{align}   
The lower bound in \eqref{eq:CodCapTreeLBCopy} is the same for every 
Steiner tree of $\Network$. 
We will use this uniform bound to lower bound the computing capacity for 
$\Network$ with respect to $f$. 
Denote the Steiner trees of $\Network$ by $t_1,\ldots,t_T$. 
Let $\epsilon > 0$ and let $r$ denote the quantity on the right hand side of \eqref{eq:CodCapTreeLBCopy}.
On every Steiner tree $t_i$, a computing rate of at least $r-\epsilon$ is
achievable by \eqref{eq:CodCapTreeLBCopy}.
Using standard arguments for time-sharing between the different Steiner trees of the network $\Network$,  
it follows that a computing rate of at least
$(r - \epsilon) \cdot \steinerNumber{\Network}$
is achievable in $\Network$,
and by letting $\epsilon \rightarrow 0$, the result follows. 
\end{proof}
The lower bound in Theorem \ref{Th:LowerBndGeneral} can be readily computed and is sometimes tight. 
The procedure used in the proof of Theorem \ref{Th:LowerBndGeneral} may potentially be improved by maximizing the sum
\begin{equation} \label{Eq:FracSteinerModified}
\sum_{t_i \in \Tree(\Network)} \!\! u_{i} \, r_i \quad \mbox{subject to} \quad  
\begin{cases}
& \!\!\!\displaystyle u_i \ge 0 \ \ \ \forall \ t_i \in \Tree(\Network) \ ,\\ 
& \!\!\!\displaystyle \sum_{i \in \treeIndex{e}} u_{i} \le 1 \ \ \ \forall \ e \in \edges(G)
\end{cases}
\end{equation}
where $r_i$ is any achievable rate%
\footnote{From Theorem \ref{Th:codCapTree}, 
$r_i$ can be arbitrarily close to \cut{t_i,f}.} 
for computing $f$ in the Steiner tree network 
$\Network_i = (t_i, \sources, \receiver)$.

We now obtain a different lower bound on the computing capacity 
in the special case when the target function is the arithmetic sum.
This lower bound is then used 
to give an alternative lower bound (in Theorem~\ref{Th:SymmFuncCodCap}) 
on the computing  capacity for the class of symmetric target functions.
The bound obtained in Theorem~\ref{Th:SymmFuncCodCap} is sometimes better than that of 
Theorem~\ref{Th:LowerBndGeneral},
and sometimes worse (Example~\ref{Ex:boundCompare} illustrates instances of 
both cases).
\begin{theorem}\label{Th:computingSum}
If $\Network$ is a network with alphabet $\alphabet = \{0, 1,\ldots,q-1\}$ and 
the arithmetic sum target function $f$, then
\begin{align*}
\codCap{\Network, f} 
&\ge
\displaystyle \underset{ C \in \cuts{\Network} }{\min} \ \frac{ \card{C}}{\log_{q} P_{q,s} } 
\end{align*}
where $P_{q,s}$ denotes the smallest prime number greater than $s(q-1)$. 
%
\end{theorem}
\begin{proof}
Let $p = P_{q,s}$
and let $\Network'$ denote the same network as $\Network$
but whose alphabet is $\field{p}$, the finite field  of order $p$.

Let $\epsilon > 0$. 
From Theorem~\ref{Th:ModuloSumCodCap}, 
there exists a $(k,n)$ solution for computing the $\field{p}$-sum of the source messages 
in $\Network'$ with an achievable computing rate satisfying
$$
\frac{k}{n} \ge \displaystyle \underset{ C \in \cuts{\Network} }{\min} \card{C} - \epsilon.
$$
%
This $(k,n)$ solution can be repeated to 
derive a $(ck,cn)$ solution for any integer $c \ge 1$
(note that edges in the network $\Network$ carry symbols from the alphabet 
$\alphabet = \{0,1,\ldots,q-1\}$, 
while those in the network $\Network'$ carry symbols from a larger alphabet $\field{p}$).
Any $(ck,cn)$ solution for computing the $\field{p}$-sum in 
$\Network'$ can be `simulated' 
in the network $\Network$ by a $\left(ck,\lceil cn \log_{q}p \rceil\right)$ code
(e.g. see \cite{ReverseButterfly}).
Furthermore, 
since $p \ge s(q-1)+1$ and the source alphabet is $\{0,1,\ldots,q-1\}$, 
the $\field{p}$-sum of the source messages in 
network $\Network$ is equal to their arithmetic sum. 
Thus, by choosing $c$ large enough, 
the arithmetic sum target function is computed in 
$\Network$ with an achievable computing rate of at least 
$$
\displaystyle \frac{\underset{ C \in \cuts{\Network} }{\min} \card{C}}{\log_{q}p } - 2\epsilon.
$$
Since $\epsilon$ is arbitrary, the result follows. 
\end{proof}
\begin{theorem} \label{Th:SymmFuncCodCap}
If $\Network$ is a network with alphabet 
$\alphabet = \{0, 1, \ldots, q-1 \}$ and a symmetric target function 
$f$, then
$$
\codCap{\Network, f} \ge 
\frac{ \displaystyle\underset{ C \in \cuts{\Network}}{\min} \card{C}}{(q - 1) \cdot 
\log_{q}P(s)} 
$$
where $P(s)$ is the smallest prime number%
\footnote{ From Bertrand's Postulate \cite[p.343]{Hardy79}, we have $P(s) \le 2s$.} 
greater than $s$. 
\end{theorem}
%
%
\begin{proof}
From Definition~\ref{Defn:SymmetricFunc}, it suffices to evaluate the histogram target function $\hat{f}$ for computing $f$.
For any set of source messages $(x_1,x_2,\ldots,x_{\cardsources}) \in \alphabet^{\cardsources}$, we have
$$
\hat{f}\left(x_1,\ldots,x_{\cardsources}\right) = \left(c_0,c_1,\ldots,c_{q - 1}\right)
$$
where $c_i = \card{\left\{j : x_j = i\right\}}$ for each $i \in \alphabet$. 
Consider the network $\Network'= (G, \sources, \receiver)$ with alphabet $\alphabet' = \{0, 1\}$. 
Then for each $i \in \alphabet$, $c_i$ can be evaluated by computing the arithmetic sum target function in 
$\Network'$ where every source node $\source_j$ is assigned the message $1$ if $x_j = i$, and $0$ otherwise. 
Since we know that 
$$
\sum_{i = 0}^{q-1} c_i = \cardsources
$$
the histogram target function $\hat{f}$ can be evaluated by computing the arithmetic sum target function $q-1$ 
times in the network $\Network'$ with alphabet $\alphabet' = \{0, 1\}$. 
Let $\epsilon > 0$. 
From Theorem~\ref{Th:computingSum} in the Appendix,
there exists a $(k,n)$ solution for computing the arithmetic sum target function in 
$\Network'$ with an achievable computing rate of at least 
$$
\frac{k}{n} \ge \frac{ \displaystyle\underset{ C \in \cuts{\Network}}{\min} \card{C}}{\log_{2}P(s)} - \epsilon .
$$
The above $(k,n)$ solution can be repeated to 
derive a $(ck,cn)$ solution for any integer $c \ge 1$. 
Note that edges in the network $\Network$ carry symbols from the alphabet 
$\alphabet = \{0,1,\ldots,q-1\}$, while those in the network $\Network'$ carry symbols from $\alphabet' = \{0, 1\}$.
Any $(ck,cn)$ code for computing the arithmetic sum function in $\Network'$ 
can be simulated in the network $\Network$ by a 
$(ck, \lceil cn \log_{q} 2\rceil)$ code%
\footnote{To see details of such a simulation, 
we refer the interested reader to  \cite{ReverseButterfly}.
}.
Thus by choosing $c$ large enough, 
the above-mentioned code can be simulated in the network $\Network$ to derive a solution for 
computing the histogram target function $\hat{f}$ with an achievable computing rate%
\footnote{Theorem~\ref{Th:SymmFuncCodCap} provides a uniform lower bound on the 
achievable computing rate for any symmetric function. 
Better lower bounds can be found by considering specific functions;
for example Theorem~\ref{Th:computingSum} 
gives a better bound for the arithmetic sum target function.} 
of at least 
$$
\frac{1}{(q-1)} \cdot \frac{1}{\log_{q}2} \cdot 
\frac{ \displaystyle\underset{ C \in \cuts{\Network}}{\min} \card{C}}{\log_{2}P(s)} - 2\epsilon .
$$
Since $\epsilon$ is arbitrary, the result follows. 
\end{proof}
\remove{
\begin{figure}[ht]  
    \centering
    \psfrag{s1}{\mbox{\Large $\source_1$}}
    \psfrag{s2}{\mbox{\Large$\source_2$}}
    \psfrag{T}{\mbox{\Large$\receiver$}}
    \psfrag{N_1}{\mbox{\Large$\Network_2$}}
    \psfrag{s3}{\mbox{\Large$\source_3$}}
    \psfrag{T}{\mbox{\Large$\receiver$}}
    \psfrag{N_2}{\mbox{\Large$\Network_3$}}
\scalebox{.5}{ \includegraphics{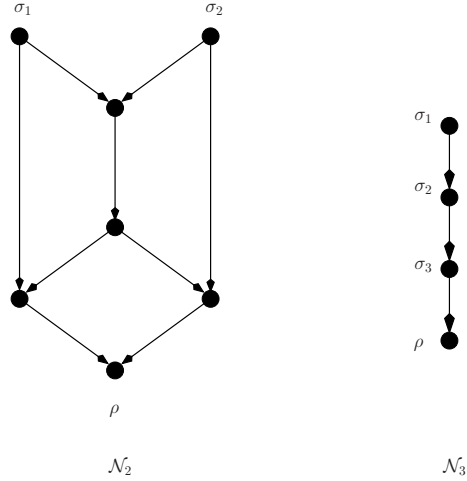}}
%
\caption{The Reverse Butterfly Network $\Network_2$ has two binary sources $\{\source_1, \source_2\}$ 
and network $\Network_3$ has three binary sources $\{\source_1, \source_2, \source_3\}$, each with $\alphabet = \{0,1\}$. 
Each network's receiver $\receiver$ computes the arithmetic sum of the source messages.}
\label{Fig:ButterflyLine}
\end{figure}
}
\begin{example} \label{Ex:boundCompare}
Consider networks $\Network_2$ and $\Network_3$ 
in Figure~\ref{Fig:ButterflyLine},
each with alphabet $\alphabet = \{0, 1\}$ 
and the (symmetric) arithmetic sum target function $f$.
Theorem~\ref{Th:SymmFuncCodCap} provides a larger lower bound on the computing capacity 
$\codCap{\Network_2, f}$ than Theorem~\ref{Th:LowerBndGeneral}, 
but a smaller lower bound on $\codCap{\Network_3, f}$.
\begin{itemize}
\item
For network $\Network_2$ (in Figure~\ref{Fig:ButterflyLine}),
we have
$\underset{C \in \cuts{\Network}}{\max} \maxRangeB{C,f}=3$
and
$\underset{C \in \cuts{\Network}}{\min} |C|=2$, 
both of which occur, for example,
when $C$ consists of the two in-edges to the receiver $\receiver$.
Also, $(q-1)\log_q P(s,q) = \log_2 3$
and $\steinerNumber{\Network} = 3/2$, 
so
\begin{align}
\codCap{\Network_2, f} &\ge (3/2)/\log_2 3 & \Comment{Theorem~\ref{Th:LowerBndGeneral}} \nonumber \\ 
\codCap{\Network_2, f} &\ge 2/\log_2 3     & \Comment{Theorem~\ref{Th:SymmFuncCodCap}}. \label{Eq:lbN2}
\end{align}
In fact, we get the upper bound $\codCap{\Network_2,f} \le 2/\log_2 3$ from Theorem~\ref{Th:upperBoundOnCodingCapacity}, and thus from \eqref{Eq:lbN2}, 
$\codCap{\Network_2,f} = 2/\log_2 3$.
\item
For network $\Network_3$,
we have
$\underset{C \in \cuts{\Network}}{\max} \maxRangeB{C,f}=4$
and
$\underset{C \in \cuts{\Network}}{\min} |C|=1$, 
both of which occur when $C = \{(\source_3,\receiver)\}$.
Also, $(q-1)\log_q P(s,q) = \log_2 5$
and $\steinerNumber{\Network} = 1$, 
so
\begin{align*}
\codCap{\Network_3, f} &\ge  1/\log_2 4 & \Comment{Theorem~\ref{Th:LowerBndGeneral}} \\
\codCap{\Network_3, f} &\ge  1/\log_2 5 & \Comment{Theorem~\ref{Th:SymmFuncCodCap}}. 
\end{align*}
From Theorem~\ref{Th:codCapTree}, we have $\codCap{\Network_3, f} = 1/\log_2 4$.
\end{itemize}
\end{example}
\begin{remark}
An open question, 
pointed out in \cite{Cannons-Dougherty-Freiling-Zeger05}, 
is whether the coding capacity of a network
can be irrational.
Like the coding capacity,
the computing capacity is the supremum of 
ratios $k/n$ for which a $(k,n)$ solution exists.
Example~\ref{Ex:boundCompare}
demonstrates that the computing capacity
of a network 
(e.g. $\Network_2$)
with unit capacity links can be 
irrational when the target function 
is the arithmetic sum function. 
\end{remark}
%
%
%
\section{On the tightness of the min-cut upper bound}
\label{Sec:Tightness}
In the previous section, 
Theorems~\ref{Th:AchievabilityMinCutIdentity} - \ref{Th:codCapTree} 
demonstrated three special instances for which the \textup{\cut{\Network,f}} upper bound is tight. 
In this section, 
we use Theorem~\ref{Th:LowerBndGeneral} and Theorem~\ref{Th:SymmFuncCodCap} 
to establish further results on the tightness of the \textup{\cut{\Network,f}} upper bound for different classes of target functions.

The following lemma provides a bound on the \footprintsize $\maxRangeB{I,f}$ for any divisible target function $f$.
\begin{lemma}
\label{Lemma:DivisibleFnMaxRange}
For any divisible target function 
$f  :  \alphabet^{\cardsources} \longrightarrow \decodeAlphabet$
and any index set $I \subseteq \{1, 2,\ldots, \cardsources\}$, 
the \footprintsize satisfies
$$
\maxRangeB{I,f} \le \card{ f\left(\alphabet^{\cardsources} \right) }.
$$
\end{lemma}
\begin{proof}
From the definition of a divisible target function, 
for any $I \subseteq \{1, 2,\ldots, \cardsources\}$, 
there exist maps $f^{I}$, $f^{I^{c}}$, and $g$ such that 
$$
f(x) = g\left( f^{I}(\VecComp{x}{I}), f^{I^{c}}(\VecComp{x}{I^{c}}) \right) \ \ \ \  \forall \ \setOfMessages \in \alphabet^{\cardsources}
$$
where $I^{c} = \{1,2,\ldots,\cardsources\} - I$. 
From the definition of the equivalence relation $\equiv$ (see Definition~\ref{Defn:EquivRelation}), 
it follows that $a, b \in \alphabet^{\card{I}}$ belong to the same equivalence class whenever
$f^{I}(a) = f^{I}(b)$. 
This fact implies that $\maxRangeB{I,f} \le \card{ f^{I}\left(\alphabet^{\card{I}}\right) }$.
We need $\card{ f^{I}\left(\alphabet^{\card{I}}\right) }  \le  \card{ f\left(\alphabet^{\cardsources} \right) }$ to complete the proof 
which follows from Definition~\ref{Defn:DivisibleFunc}(2).
\end{proof}
\begin{theorem}\label{Th:DivFuncCodCap}
If $\Network$ is a network with a divisible target function 
$f$, then
$$
\codCap{\Network,f} \ge \frac{\steinerNumber{\Network}}{\card{ \inEdges{\receiver} }} \cdot \mbox{\textup{\cut{\Network,f}}} 
$$
where $\inEdges{\receiver}$ denotes the set of in-edges of the receiver $\receiver$.
\label{thm:divisible}
\end{theorem}
\begin{proof}
Let $\alphabet$ be the network alphabet. 
From Theorem~\ref{Th:LowerBndGeneral}, 
\begin{align}
\codCap{\Network,f} & \ge \steinerNumber{\Network} \cdot \underset{C \in \cuts{\Network} }{\min} \ 
 \frac{1}{\log_{\card{\alphabet}} \maxRangeB{C,f}} \nonumber \\
\label{Eq:LowerBndCodCap}
&  \ge \steinerNumber{\Network} \cdot \frac{1}{\log_{\card{\alphabet}} \card{ f\left(\alphabet^{\cardsources}\right) }} 
          & \Comment{Lemma~\ref{Lemma:DivisibleFnMaxRange}}.
\end{align}
On the other hand, 
for any network $\Network$, 
the set of edges $\inEdges{\receiver}$ 
is a cut that separates the set of sources $\sources$ from $\receiver$.
Thus,
\begin{align}
\nonumber
\mbox{\cut{\Network,f}} &\le \frac{\card{\inEdges{\receiver}}}{\log_{\card{\alphabet}}\maxRangeB{\inEdges{\receiver},f}} 
       & \Comment{\eqref{Eq:min-f-cut}} \\
\label{Eq:UpperBndMinCut}
& = \frac{\card{\inEdges{\receiver}}}{\log_{\card{\alphabet}} \card{f\left(\alphabet^{\cardsources}\right)} } 
       & \Comment{$I_{ \inEdges{\receiver} } = \sources$ and Definition~\ref{Defn:NoOfEquivClasses} }.
\end{align} 
%
%
Combining \eqref{Eq:LowerBndCodCap} and \eqref{Eq:UpperBndMinCut} completes the proof.
\end{proof}
\begin{theorem}\label{Th:SymmFuncGap}
If $\Network$ is a network with alphabet $\alphabet = \{0, 1, \ldots, q-1 \}$ and symmetric target function $f$, then
$$
\codCap{\Network,f} \ge \frac{\log_q \hat{\maxRangeB{}}_f}{(q - 1) \cdot \log_{q}P(s)} \cdot \mbox{\textup{\cut{\Network,f}}} 
$$
where $P(s)$ is the smallest prime number greater than $s$ and%
\footnote{From our assumption, $\hat{\maxRangeB{}}_f \ge 2$ for any target function $f$.} 
$$
\hat{\maxRangeB{}}_f = \underset{ I \subseteq \{1,\ldots,\cardsources\} }{\min} \maxRangeB{I, f} .
$$
\end{theorem}
\begin{proof}
The result follows immediately from Theorem~\ref{Th:SymmFuncCodCap} 
and since for any network $\Network$ and any target function $f$, 
\begin{align*}
\mbox{\cut{\Network,f}} \le \frac{1}{\log_q \hat{\maxRangeB{}}_f} \ \displaystyle\underset{ C \in \cuts{\Network}}{\min} \card{C}  &\Comment{\eqref{Eq:min-f-cut} and the definition of $\hat{\maxRangeB{}}_f$}.
\end{align*}
\end{proof}
The following results provide bounds on the gap between the computing capacity and the min-cut for $\lambda$-exponential and $\lambda$-bounded functions 
(see Definition~\ref{Defn:LambdaExpFunc}).
\begin{theorem}\label{Th:GapLambdaExp}
If $\lambda \in (0, 1]$ 
and $\Network$ is a network with a $\lambda$-exponential target function 
$f$, then
$$
\codCap{\Network,f} \ge \lambda \cdot \mbox{\textup{\cut{\Network,f}}} .
$$
\end{theorem}
\begin{proof}
We have %
\begin{align*}
\mbox{\cut{\Network,f}} & = \underset{C \in \cuts{\Network}}{\min} 
 \ \frac{\card{C}}{\log_{\card{\alphabet}} \maxRangeB{C,f}} \\
                        &  \le \underset{C \in \cuts{\Network}}{\min} \ \frac{\card{C}}{\lambda  \card{I_C}} 
                           & \Comment{$f$ being $\lambda$-exponential} \\
                 & = \frac{1}{\lambda} \cdot \mbox{\minCut{\Network}} & \Comment{\eqref{Eq:min-cut}} .
\end{align*}
Therefore, 
\begin{align*}
\hspace{-1in}
\frac{\mbox{\cut{\Network,f}}}{\codCap{\Network,f}} &\le  \frac{1}{\lambda} \cdot \frac{\mbox{\minCut{\Network}}}{\codCap{\Network,f}} \hspace{1.55in} \\
&\le \frac{1}{\lambda} 
\end{align*}
where the last inequality follows because a computing rate of
\minCut{\Network} is achievable for the identity target function
from Theorem~\ref{Th:AchievabilityMinCutIdentity}, 
and the computing capacity for any target function $f$ is lower bounded by the computing capacity for the identity target function
(since any target function can be computed from the identity function), i.e., $\codCap{\Network,f} \ge \mbox{\minCut{\Network}}$.
\end{proof}
\begin{theorem}\label{Th:GapLambdaBded}
Let $\lambda > 0$. If $\Network$ is a network with alphabet $\alphabet$ and a $\lambda$-bounded target function $f$,
and all non-receiver nodes in the network $\Network$ are sources, then
$$
\codCap{\Network,f} \ge \frac{\log_{\card{\alphabet}} \hat{\maxRangeB{}}_f }{\lambda} \cdot
\mbox{\textup{\cut{\Network,f}}} 
$$
where
$$
\hat{\maxRangeB{}}_f = \underset{ I \subseteq \{1,\ldots,\cardsources\} }{\min} \maxRangeB{I, f} .
$$
\end{theorem}
\begin{proof}
For any network $\Network$ such that all non-receiver nodes are sources, 
it follows from Edmond's Theorem \cite[p.405, Theorem 8.4.20]{West01} that   
$$
\steinerNumber{\Network} = \underset{C \in \cuts{\Network}}{\min} \ \card{C} .
$$
Then, 
\begin{align} 
\codCap{\Network,f} &\ge \underset{C \in \cuts{\Network}}{\min} \ \card{C} \cdot 
 \underset{C \in \cuts{\Network} }{\min} \ \frac{1}{\log_{\card{\alphabet}} \maxRangeB{C,f}} & \Comment{Theorem~\ref{Th:LowerBndGeneral}} \nonumber \\
\label{eq:cc}
& \ge \underset{C \in \cuts{\Network}}{\min} \ \frac{\card{C}}{\lambda} & \Comment{$f$ being $\lambda$-bounded}.
\end{align}
On the other hand, 
\begin{align} \label{eq:cn}
\mbox{\cut{\Network,f}} & = \underset{C \in \cuts{\Network}}{\min} 
\ \frac{\card{C}}{\log_{\card{\alphabet}} \maxRangeB{C,f}}  \nonumber \\
& \le \underset{C \in \cuts{\Network}}{\min} 
 \ \frac{\card{C}}{\log_{\card{\alphabet}} \hat{\maxRangeB{}}_f } & \Comment{the definition of $\hat{\maxRangeB{}}_f$}. 
 \end{align}
Combining \eqref{eq:cc} and \eqref{eq:cn} gives
\begin{align*}
\frac{ \mbox{ \cut{\Network,f} } }{\codCap{\Network,f}} 
                              & \le  \underset{C \in \cuts{\Network}}{\min} 
\ \frac{\card{C}}{\log_{\card{\alphabet}} \hat{\maxRangeB{}}_f } \cdot \frac{1}{\underset{C \in \cuts{\Network}}{\min}  \frac{\card{C}}{\lambda} } \\
                              & =  \frac{\lambda}{\log_{\card{\alphabet}}\hat{\maxRangeB{}}_f } .
\end{align*}
\end{proof}
Since the maximum and minimum functions
are $1$-bounded,
and $\hat{\maxRangeB{}}_f = \card{\alphabet}$ for each,
we get the following corollary.
\begin{corollary}
Let $\alphabet$ be any ordered alphabet and
let $\Network$ be any network such that all non-receiver
nodes in the network  are sources.
If the target function $f$ is either the maximum or the minimum function,
then
$$
\codCap{\Network,f} = \mbox{ \textup{\cut{\Network,f}}} .
$$
\label{cor:maxmin}
\end{corollary}
Theorems~\ref{Th:DivFuncCodCap} - \ref{Th:GapLambdaBded} 
provide bounds on the tightness of the \textup{\cut{\Network,f}} upper bound for different classes of target functions. 
In particular, we show that for $\lambda$-exponential 
(respectively, $\lambda$-bounded) target functions, the computing capacity $\codCap{\Network,f}$ 
is at least a constant fraction of the \textup{\cut{\Network,f}} for any constant $\lambda$ and any network 
$\Network$ (respectively, any network $\Network$ where all non-receiver nodes are sources). 
The following theorem shows by means of an example target function $f$ and a network 
$\Network$, that the \textup{\cut{\Network,f}} upper bound cannot always 
approximate the computing capacity $\codCap{\Network,f}$ up to a constant fraction. 
Similar results are known in network coding as well as in multicommodity flow. 
It was shown in~\cite{LeightonRao99} that
when $s$ source nodes communicate
independently with the same number of receiver nodes, 
there exist networks whose maximum multicommodity flow is 
$O(1/ \log s)$ times a well known cut-based upper bound. 
It was shown in~\cite{Harvey} that 
with network coding
there exist networks whose maximum throughput is $O(1/ \log s)$ times the best known cut
bound (i.e. "meagerness"). 
Whereas these results do not hold for single-receiver networks 
(by Theorem~\ref{Th:AchievabilityMinCutIdentity}),
the following similar bound holds for network computing in single-receiver networks.
The proof of Theorem~\ref{Th:UnboundedGap} uses Lemma~\ref{lem:capacityLimitOf_G_N_L} which is presented 
in the Appendix. 
\begin{theorem}\label{Th:UnboundedGap} \label{Th:HarveyGap}
For any $\epsilon > 0$, there exist networks $\Network$ such that for the arithmetic sum target function 
$f$,
$$
\codCap{\Network,f} = O \! \left( \frac{1}{(\log s )^{1-\epsilon}} \right)  \cdot \mbox{\textup{\cut{\Network,f}}}.
$$
\end{theorem}
\begin{proof}
Note that for the network $\Network_{\mbox{\tiny{$M$,$L$}}}$ and the arithmetic sum target function $f$, %
\begin{align*}
\mbox{\cut{\Network_{\mbox{\tiny{$M$,$L$}}},f}} 
  &= \underset{C \in \cuts{\Network_{\mbox{\tiny{$M$,$L$}}}}}{\min} \ \frac{\card{C}}{\log_{2}\left(\card{I_C} + 1\right)} 
     & \Comment{\eqref{Eq:min-cutSum}}. 
\end{align*}
Let $m$ be the number of sources disconnected from the receiver $\receiver$ by a cut $C$ in the network $\Network_{\mbox{\tiny{$M$,$L$}}}$. 
For each such source $\source$, 
the cut $C$ must contain the edge $(\source, \receiver)$ as well as either the $L$ parallel 
edges $(\source, \source_0)$ or the $L$ parallel edges $(\source_0, \receiver)$. 
Thus, 
\begin{align}
\label{Eq:MinCutNML}
\mbox{\cut{\Network_{\mbox{\tiny{$M$,$L$}}},f}}   
&= \min_{1 \leq m \leq M} \left\{ \frac{L + m}{\log_2(m+1)} \right\}.
\end{align}
Let $m^{*}$ attain the minimum in \eqref{Eq:MinCutNML}
and define $c^{*}$= \cut{\Network_{\mbox{\tiny{$M$,$L$}}},f}.
Then,
\begin{align}
c^{*}/ \ln 2 
&\ge \min_{1 \le m \le M} \left\{ \frac{m+1}{\ln (m+1)} \right\} 
\ge \min_{x \ge 2} \left\{ \frac{x}{\ln x} \right\} 
> \min_{x \ge 2} \left\{ \frac{x}{x-1}\right\} 
> 1 \nonumber \\ 
%
\nonumber
L &= c^{*} \log_{2}\left(m^{*}+1\right) - m^{*} & \Comment{\eqref{Eq:MinCutNML}}\\
   &\le c^{*} \log_{2}\left( \frac{c^{*}}{\ln 2} \right) -  \left(\frac{c^{*}}{\ln 2} - 1\right) \label{Eq:SecIneq}
\NOPROCESS{
   \\
   & \le c^{*} \log_{2}\left( \frac{c^{*}}{\ln 2} \right) & \Comment{\eqref{eq:100},\eqref{Eq:SecIneq}} \label{eq:101}}
\end{align}
where \eqref{Eq:SecIneq} follows since the function $c^{*}\log_{2}\left(x+1\right) - x$ 
attains its maximum value over $(0, \infty)$ at $x = (c^{*} / \ln 2) - 1$. 
Let us choose $L = \lceil (\log M)^{1-(\epsilon/2)} \rceil $.
We have
\begin{align} 
L & = O \! \left( \mbox{\cut{\Network_{\mbox{\tiny{$M$,$L$}}},f}}  
\log_{2}(\mbox{\cut{\Network_{\mbox{\tiny{$M$,$L$}}},f}} )  \right) & \Comment{\eqref{Eq:SecIneq}} \label{Eq:newMinCutBnd}\\
\mbox{\cut{\Network_{\mbox{\tiny{$M$,$L$}}},f}} & = \Omega((\log M)^{1-\epsilon}) \label{Eq:newMinCutBnd1}
&\Comment{\eqref{Eq:newMinCutBnd}}\\
\codCap{\Network_{M,L},f} 
 & = O(1)  & \Comment{Lemma~\ref{lem:capacityLimitOf_G_N_L}} \notag\\
 & = O \! \left( \frac{1}{(\log M)^{1-\epsilon}} \right) \cdot \mbox{\cut{\Network_{\mbox{\tiny{$M$,$L$}}},f}}  
& \Comment{\eqref{Eq:newMinCutBnd1}}.\notag
\end{align}
\end{proof}

\section{An example network} 
\label{Sec:Example}
\begin{figure}[ht]
\begin{center}
\psfrag{X}{$\source_3$}
\psfrag{Y}{$\source_1$}
\psfrag{Z}{$\source_2$}
\psfrag{T}{$\receiver$}
\scalebox{.7}{\includegraphics{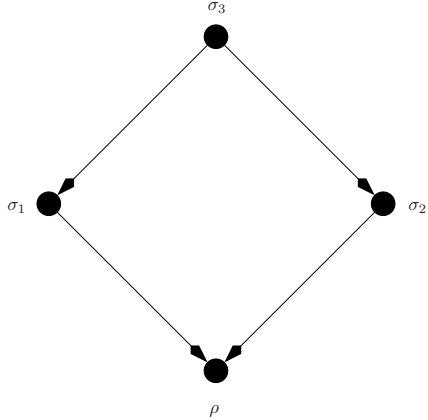}}
\end{center}
\caption{Network $\hat{\Network}$ has three binary sources, $\source_1$, $\source_2$, and $\source_3$ with $\alphabet = \{0,1\}$ 
and the receiver $\receiver$ computes the arithmetic sum of the source messages. 
} 
\label{Fig:diamondNetwork}
\end{figure}
%
%
In this section, we evaluate the computing capacity for an 
example network and a target function (which is divisible and symmetric) 
and show that the min-cut bound is not tight. 
In addition, the example demonstrates that the lower bounds
discussed in Section~\ref{Sec:LowerBounds} are not always tight and illustrates
the combinatorial nature of the computing problem.
\begin{theorem}\label{Th:DiamondNetworkCapacity}
The computing capacity of network $\hat{\Network}$ with respect to the arithmetic sum target function $f$ is
$$
\codCap{\hat{\Network},f}=\frac{2}{1+ \log_{2}3}.
$$
\end{theorem}
\begin{proof}
For any $(k, n)$ solution for computing $f$,
let $w^{(1)}, w^{(2)}, w^{(3)}\in\{0,1\}^k$
denote the message vectors generated by sources
$\source_1$, $\source_2$, $\source_3$, respectively,
and let $z_1, z_2 \in \{0,1\}^n$ be the
vectors carried by edges
$(\source_1,\receiver)$ and $(\source_2,\receiver)$,
respectively.

Consider any positive integers $k, n$ such that $k$ is even and 
\begin{equation}
\label{Eq:ProposedRate}
\frac{k}{n} \le \frac{2}{1+ \log_{2}3} .
\end{equation}
Then we have 
\begin{equation}
\label{Eq:Mapping}
2^{n} \geq 3^{k/2}  2^{k/2} .
\end{equation}
We will describe
a $(k, n)$ network code for computing $f$ in the network $\hat{\Network}$.
Define vectors $y^{(1)}, y^{(2)} \in\{0,1\}^k$ by:
\begin{align*}
y^{(1)}_i &= \left\{
 \begin{array}{cc}
   \VecComp{w^{(1)}}{i} + \VecComp{w^{(3)}}{i} & \mbox{ if } 1 \le i \le k/2\\
   \VecComp{w^{(1)}}{i}                        & \mbox{ if } k/2 \le i \le k
 \end{array}
\right. \\
y^{(2)}_i &= \left\{
 \begin{array}{cc}
   \VecComp{w^{(2)}}{i}                        & \mbox{ if } 1 \le i \le k/2\\
   \VecComp{w^{(2)}}{i} + \VecComp{w^{(3)}}{i} & \mbox{ if } k/2 \le i \le k .
 \end{array}
\right. 
\end{align*}
The first $k/2$ components of $y^{(1)}$ can take on the values $0,1,2$, and
the last  $k/2$ components can take on the values $0,1$,
so there are a total of $3^{k/2}2^{k/2}$ possible values for $y^{(1)}$,
and similarly for $y^{(2)}$. 
From \eqref{Eq:Mapping}, 
there exists a mapping that assigns unique values to $z_1$
for each different possible value of $y^{(1)}$,
and similarly for $z_2$ and $y^{(2)}$.
This induces a solution for $\hat{\Network}$ as summarized below.

The source $\source_3$ sends its full message vector $w^{(3)}$ ($k < n$) to each of the two nodes
it is connected to.
Source $\source_1$ 
(respectively, $\source_2$)
computes the vector $y^{(1)}$ 
(respectively, $y^{(2)}$), 
then computes the vector $z_1$ 
(respectively, $z_2$), 
and finally sends $z_1$ 
(respectively, $z_2$) on its out-edge. 
The receiver $\receiver$ determines $y^{(1)}$ and $y^{(2)}$ from $z_1$ and $z_2$, respectively,
and then computes $y^{(1)} + y^{(2)}$,
whose $i$-th component is 
$\VecComp{w^{(1)}}{i} + \VecComp{w^{(2)}}{i} + \VecComp{w^{(3)}}{i}$, i.e., 
the arithmetic sum target function $f$. 
The above solution achieves a computing rate of $k/n$. 
From \eqref{Eq:ProposedRate}, 
it follows that 
\begin{equation}
\label{Eq:ExampleLB}
\codCap{\hat{\Network},f} \ge \frac{2}{1+ \log_{2}3} .
\end{equation}
%

We now prove a matching upper bound on the computing capacity 
$\codCap{\hat{\Network},f}$. 
Consider any $(k, n)$ solution for computing the arithmetic sum target function $f$ in network $\hat{\Network}$. 
For any $\sumVec \, \in \, \{0,1,2,3\}^{k}$, 
let  
\begin{align*}
A_p &= \{ (z_1,z_2): w^{(1)}+ w^{(2)}+ w^{(3)}= p \}.
\end{align*}
That is, each element of $A_p$ is a possible pair of input edge-vectors to the receiver when
the target function value equals $p$.

%

Let $j$ denote the number of components of $\sumVec$ that are either $0$ or $3$.
Without loss of generality, 
suppose the first $j$ components of $\sumVec$ belong to $\{0, 3\}$ and
define
$\tilde{w}^{(3)} \in \{0,1\}^k$ by
$$
\tilde{w}^{(3)}_i = \left\{
  \begin{array}{cc}
    0 & \mbox{ if } p_i \in \{0,1\} \\
    1 & \mbox{ if } p_i \in \{2,3\} .
  \end{array}
\right.
$$
Let
$$
T =  \{ (w^{(1)},w^{(2)}) \, \in \, \{0,1\}^k \times \{0,1\}^k \; : \;  w^{(1)}+ w^{(2)} + \tilde{w}^{(3)} = \sumVec  \}
$$
and notice that
\begin{align} \label{eq:seth}
\left\{ (z_1,z_2) : (w^{(1)}, w^{(2)}) \in T, w^{(3)} = \tilde{w}^{(3)} \right\} \subseteq A_p.
\end{align}
If $w^{(1)}+ w^{(2)} + \tilde{w}^{(3)}=p$,
then:
\begin{itemize}
\itemsep-0.3em
\item [(i)] $p_i - \tilde{w}^{(3)}_i =0$ implies $w^{(1)}_i = w^{(2)}_i = 0$;
\item [(ii)]  $p_i - \tilde{w}^{(3)}_i =2$ implies $w^{(1)}_i = w^{(2)}_i = 1$;
\item [(iii)] $p_i - \tilde{w}^{(3)}_i =1$ implies $(w^{(1)}_i,w^{(2)}_i) = (0,1)$ or $(1,0)$.
\end{itemize}
Thus,
the elements of $T$ consist of $k$-bit vector pairs $(w^{(1)}, w^{(2)})$
whose first $j$ components are fixed and equal
(i.e., both are $0$ when $p_i=0$ and both are $1$ when $p_i=3$),
and whose remaining $k-j$ components can each be chosen from two possibilities
(i.e., either $(0,1)$ or $(1,0)$,  when $p_i\in\{1,2\}$).
This observation implies that

\begin{align} 
\card{T} = 2^{k-j}.
\label{eq:8}
\end{align}
Notice that if only $w^{(1)}$ changes,
then the sum $w^{(1)}+ w^{(2)}+ w^{(3)}$ changes,
and so $z_1$ must change 
(since $z_2$ is not a function of $w^{(1)}$)
in order for the receiver to compute the target function.
Thus, if $w^{(1)}$ changes and $w^{(3)}$ does not change,
then $z_1$ must still change, regardless of whether $w^{(2)}$ changes or not.
More generally,
if the pair $(w^{(1)},w^{(2)})$ changes,
then the pair $(z_1,z_2)$ must change.
Thus,
\begin{align} 
\card{\left\{ (z_1,z_2) : (w^{(1)}, w^{(2)}) \in T, w^{(3)} = \tilde{w}^{(3)} \right\}} 
&\ge \card{T}
\label{eq:11}
\end{align}
and therefore
\begin{align}
\card{A_p} 
&\ge \card{\left\{ (z_1,z_2) : (w^{(1)}, w^{(2)}) \in T, w^{(3)} = \tilde{w}^{(3)} \right\}} &\Comment{\eqref{eq:seth}}\notag\\
&\ge \card{T}                                                     &\Comment{\eqref{eq:11}}\notag\\
& = 2^{k-j}.                                                      &\Comment{\eqref{eq:8}}
\label{eq:12}
\end{align}

We have the following inequalities:
\begin{align}
4^{n} 
&\geq \card{\{(z_1,z_2): w^{(1)}, w^{(2)}, w^{(3)} \in \{0,1\}^k\}} \notag\\
&= \sum_{p \in \{0,1,2,3\}^k} \card{A_p} \label{eq:7}\\
&= \sum_{j=0}^k \ 
   \sum_{\substack{{p \in \{0,1,2,3\}^k} \\ {\card{ \{i : p_i \in \{0,3\} \} } = j} }}
   \card{A_p}\notag\\
& \ge \sum_{j=0}^k \ 
   \sum_{\substack{{p \in \{0,1,2,3\}^k} \\ {\card{ \{i : p_i \in \{0,3\} \} } = j} }}
   2^{k-j} & \Comment{\eqref{eq:12}}  \notag\\
&= \sum_{j=0}^k \binom{k}{j}  2^k 2^{k-j} \notag \\
&= 6^k \label{eq:1}
\end{align}
where \eqref{eq:7} follows since the $A_p$'s must be disjoint in order for the receiver 
to compute the target function.
Taking logarithms of both sides of \eqref{eq:1}, 
gives
$$
\frac{k}{n} \leq \frac{2}{1 + \log_{2}3}
$$
which holds for all $k$ and $n$,
and therefore
\begin{equation}
\label{Eq:ExampleUB}
\codCap{\hat{\Network},f} \le \frac{2}{1+ \log_{2}3}.
\end{equation}
Combining \eqref{Eq:ExampleLB} and \eqref{Eq:ExampleUB} concludes the proof.
\end{proof}
\begin{corollary}\label{Th:PositiveGapNetwork}
For the network $\hat{\Network}$ with the arithmetic sum target function $f$, 
$$
\codCap{\hat{\Network}, f} < \mbox{\textup{\cut{\hat{\Network},f}}} .
$$
\end{corollary}
\begin{proof}
Consider the network $\hat{\Network}$ depicted in Figure~\ref{Fig:diamondNetwork} 
with the arithmetic sum target function $f$. 
It can be shown that
the \footprintsize  $\maxRangeB{C,f} = \card{I_{C}} + 1$ for any cut $C$,
and thus
\begin{align*}
\mbox{\textup{\cut{\hat{\Network},f}}} &= 1  \Comment{\eqref{Eq:min-cutSum}} .
\end{align*}
The result then follows immediately from Theorem~\ref{Th:DiamondNetworkCapacity}. 
\end{proof}
\begin{remark}
In light of Theorem~\ref{Th:DiamondNetworkCapacity},
we compare the various lower bounds on the computing capacity 
of the network $\hat{\Network}$ derived in 
Section~\ref{Sec:LowerBounds} with the exact computing capacity.
It can be shown that $\steinerNumber{\hat{\Network}} = 1$.
If $f$ is the arithmetic sum target function, then
\begin{align*}
\codCap{\hat{\Network},f} & \ge 1/2  & \Comment{Theorem~\ref{Th:LowerBndGeneral}} \\
\codCap{\hat{\Network},f} & \ge 1/\log_2 5 & \Comment{Theorem~\ref{Th:SymmFuncCodCap}} \\
\codCap{\hat{\Network},f} & \ge 1/2  & \Comment{Theorem~\ref{thm:divisible}}.
\end{align*}
Thus, this example demonstrates that the lower bounds
obtained in Section~\ref{Sec:LowerBounds} are not always tight and illustrates the combinatorial nature of the problem.
\end{remark}
%
%
\section{Conclusions}\label{Sec:conclusions}
We examined the problem of network computing. 
The network coding problem is a special case when the function to be computed is the identity. 
We have focused on the case when a single receiver node computes a function of the source messages and have shown that 
while for the identity function the min-cut bound is known to be tight for all networks, 
a much richer set of cases arises when computing arbitrary functions, 
as the min-cut bound can range from being tight to arbitrarily loose. 
One key contribution of the paper is to show the theoretical breadth of the considered topic, 
which we hope will lead to further research. 
This work identifies target functions (most notably, the arithmetic sum function) for which the min-cut bound is not always tight 
(even up to a constant factor) and future work includes deriving more sophisticated bounds for these scenarios. 
Extensions to computing with multiple receiver nodes, 
each computing a (possibly different) function of the source messages, are of interest.

\clearpage
\section{Appendix}
%
%
%
Define the function 
\begin{align*}
Q \ & : \ \displaystyle \prod_{i=1}^{M} \{0,1\}^k \longrightarrow \{0,1,\ldots,M\}^k \nonumber
\end{align*}
as follows. 
For every $a = (a^{(1)},a^{(2)},\ldots,a^{(M)})$ such that each $a^{(i)} \in \{0,1\}^k$, 
\begin{align} \label{Eq:sumDefinition}
 \VecComp{\sumset{a}}{j} & = \sum_{i=1}^{M} \VecComp{a^{(i)}}{j} 
\quad \mbox{for every $j \in \left\{1,2,\ldots,k\right\}$}.
\end{align}
We extend $Q$ for $X \subseteq \displaystyle \prod_{i=1}^{M} \{0,1\}^k$ 
by defining $\sumset{X} = \{\sumset{a} \ : \ a \in X\}$.

We now present Lemma~\ref{lem:capacityLimitOf_G_N_L}. The proof uses Lemma~\ref{Lemma:CombinatorialBound}, which is presented thereafter. 
%
We define the following function which is used in the next lemma. 
Let
\begin{equation}
\label{Eq:CardConst}
\cardConst{x} = \entropy^{-1}\left( \frac{1}{2}\left(1-\frac{1}{x}\right) \right) \bigcap \left[0, \frac{1}{2}\right]
\ \ \ \ \mbox{for } x \ge 1
\end{equation}
where $\entropy^{-1}$ 
denotes the inverse of the
binary entropy function $\entropy(x) = -x\log_2 x -(1-x)\log_2(1-x)$.
Note that $\cardConst{x}$ is an increasing function of $x$.
\begin{lemma} \label{lem:capacityLimitOf_G_N_L}
If 
$
\displaystyle\lim_{M \rightarrow \infty} \frac{L}{\log_2 M} = 0,
$
then
$
\displaystyle\lim_{M \rightarrow \infty} \codCap{\Network_{\mbox{\tiny{$M$,$L$}}},f} = 1.
$
\end{lemma}
\begin{proof}
For any $M$ and $L$, 
a solution with computing rate $1$ is obtained by having each source $\source_i$ 
send its message directly to the receiver on the edge $(\source_i, \receiver)$.  
Hence $\codCap{\Network_{\mbox{\tiny{$M$,$L$}}},f} \ge 1$.
Now suppose that $\Network_{\mbox{\tiny{$M$,$L$}}}$ has a 
$(k,n)$ solution with computing rate $k / n > 1$ and for each $i \in \{1, 2,\ldots, M\}$, let
$$
g_{i} \; : \; \{0, 1\}^{k} \longrightarrow \{0, 1\}^{n}
$$
be the corresponding encoding function on the edge $\left(\source_{i}, \receiver\right)$. 
Then for any $A_1,A_2,\ldots,A_M \subseteq \{0,1\}^{k}$, we have 
\begin{equation}
\label{Eq:CondnSoln}
\left( \prod_{i=1}^{M} \card{g_i\left(A_i\right)} \right) \cdot 2^{nL} 
\ge \card{\sumset{\prod_{i=1}^{M} A_i}} .
\end{equation}
Each $A_i$ represents a set of possible message vectors of source $\source_i$. 
The left-hand side of \eqref{Eq:CondnSoln} is the maximum number of different possible instantiations of the information carried 
by the in-edges to the receiver $\receiver$ 
(i.e., $\card{g_i\left(A_i\right)}$ possible vectors on each edge $(\source_i, \receiver)$ and $2^{nL}$ 
possible vectors on the $L$ parallel edges $(\source_0, \receiver)$). 
The right-hand side of \eqref{Eq:CondnSoln} is the number of distinct sum vectors that the receiver needs to discriminate, 
using the information carried by its in-edges. 

For each $i \in \{1, 2,\ldots, M\}$, 
let $z_i \in \{0, 1\}^{n}$ be such that 
$\card{g_{i}^{-1}\left(z_i\right)} \ge 2^{k - n}$ 
and choose $A_i =  g_{i}^{-1}\left(z_i\right)$ for each $i$. 
Also, 
let $\NbinaryBlocks{U}=\displaystyle \prod_{i=1}^{M} A_i$.  
Then we have    
\begin{align}
\label{Eq:NeccCondn}
\card{\sumset{\NbinaryBlocks{U}}} \leq 2^{nL} & \Comment{$\card{g_i\left(A_i\right)} = 1$ and \eqref{Eq:CondnSoln}}.
\end{align}
Thus \eqref{Eq:NeccCondn} is a necessary condition for the existence of a $(k,n)$ solution for computing $f$ in the network $\Network_{\mbox{\tiny{$M$,$L$}}}$. 
Lemma~\ref{Lemma:CombinatorialBound} shows that%
\footnote{One can compare this lower bound to the upper bound 
$\card{\sumset{\NbinaryBlocks{U}}} \le (M+1)^{k}$ which follows from 
\eqref{Eq:sumDefinition}.} 
\begin{equation}
\label{Eq:NoOfSumsLB}
\card{\sumset{\NbinaryBlocks{U}}} \ge (M + 1)^{\cardConst{k/n}k }
\end{equation} 
where the function $\gamma$ is defined in \eqref{Eq:CardConst}. 
Combining \eqref{Eq:NeccCondn} and \eqref{Eq:NoOfSumsLB}, 
any $(k,n)$ solution for computing $f$ in the network $\Network_{\mbox{\tiny{$M$,$L$}}}$ with rate $r = k / n > 1$ must satisfy 
\begin{align}
\label{Eq:ExistenceSoln} 
r \ \cardConst{r} \  \log_2(M+1) \le  \frac{1}{n} \log_{2} \card{\sumset{\NbinaryBlocks{U}}} \le L .
\end{align}
From \eqref{Eq:ExistenceSoln}, we have 
\begin{align} 
r \ \cardConst{r} &  \le \frac{L}{\log_2(M+1)}. \label{Eq:rateEqn} 
\end{align}
%
The quantity $r\cardConst{r}$ is monotonic increasing from $0$ to $\infty$
on the interval $[1,\infty)$
and the right hand side of \eqref{Eq:rateEqn} goes to zero as $M \rightarrow \infty$.
Thus, the rate $r$ can be forced to be arbitrarily close to $1$ by making $M$ sufficiently large,
i.e. 
$\codCap{\Network_{\mbox{\tiny{$M$,$L$}}},f} \le 1$.
In summary,
$$
\lim_{M \longrightarrow \infty} \codCap{\Network_{\mbox{\tiny{$M$,$L$}}},f} = 1.
$$
%

\NOPROCESS{
Let $\beta>1$ and let 
$$
M_{\beta} = 2 ^{L/(\beta\cardConst{\beta})} - 1 . 
$$
Since $r\cardConst{r}$ is strictly increasing in $r$, 
we have that if $M > M_\beta$ and $r \ge \beta$, 
then
$$
r \cardConst{r} \log_2(M+1) > L.
$$
If $M > M_{\beta}$, 
then from \eqref{Eq:ExistenceSoln} no solution for computing $f$ in 
$\Network_{\mbox{\tiny{$M$,$L$}}}$ can have rate $r \ge \beta$.
In other words, if $M > M_{\epsilon}$, then  
$$
1 \le \codCap{\Network_{\mbox{\tiny{$M$,$L$}}},f} \le \beta .
$$
Since $\beta$ is arbitrary, this implies that   
$$
\lim_{M \longrightarrow \infty} \codCap{\Network_{\mbox{\tiny{$M$,$L$}}},f} = 1.
$$
}
\end{proof}
%
%
\begin{lemma} \label{Lemma:CombinatorialBound}
Let $k, n, M$ be positive integers such that $k > n$. 
For each $i \in \{1,2,\ldots,M\}$,  
let $A_i \subseteq \{0, 1\}^{k}$ 
be such that $\card{A_i} \ge 2^{k-n}$
and let $\NbinaryBlocks{U} = \displaystyle \prod_{i=1}^{M} A_i$. 
Then, 
$$
\card{\sumset{\NbinaryBlocks{U}}} \ge (M + 1)^{\cardConst{k/n}k } .
$$
\end{lemma}
\begin{proof}
The result follows from Lemmas~\ref{Lemma1} and \ref{Lemma2}.
\end{proof}
%
%

The remainder of this Appendix
is devoted to the proofs of lemmas used
in the proof of Lemma~\ref{Lemma:CombinatorialBound}. 
Before we proceed, we need to define some more notation.
For every $j \in \{1, 2,\ldots, k\}$, 
define the map
%
\begin{align*}
\zeroAt{j} \ & : \ \{0, 1, \ldots, M\}^k  \; \longrightarrow \; \{0, 1, \ldots, M\}^k  
\end{align*}
%
%
by
\begin{align}
\VecComp{\left(\zeroAt{j}(p)\right)}{i} = 
\begin{cases} \label{Eq:majFunction}
\max \left\{0, \VecComp{p}{i} -1\right\} & \mbox{if} \ i = j \\
\VecComp{p}{i} & \mbox{otherwise.} 
\end{cases}
\end{align}
That is, 
the map $\zeroAt{j}$ subtracts one from the $j$-th component of the input vector 
(as long as the result is non-negative) and leaves all the other components the same. 
For every $j \in \{1,2,\ldots,k\}$, 
define the map 
$$
\hat{\phi}^{(j)} : 2^{\{0, 1\}^k} \ \times \ \{0, 1\}^k \longrightarrow \{0, 1\}^k 
$$
by
%
\begin{align}
\label{Eq:Phi-j}
\hat{\phi}^{(j)}(A,a) & =
\begin{cases}
\zeroAt{j}(a) & \mbox{if} \ \zeroAt{j}(a) \notin A \\
a      & \mbox{otherwise}
\end{cases}
\ \ \ \ \  \forall \ A \subseteq \{0, 1\}^k \mbox{ and } a \in \{0,1\}^k.
\end{align}
Define 
$$
\phi^{(j)} : 2^{\{0, 1\}^k} \longrightarrow 2^{\{0, 1\}^k}
$$
by
\begin{equation}
\label{Eq:DefnPhi}
\phi^{(j)}(A) = \left\{ \hat{\phi}^{(j)}(A,a) \ : \ a \ \in \ A \right\}.
\end{equation}
Note that 
\begin{equation}
\label{Eq:TransformedCard}
\card{\phi^{(j)}(A)} = \card{A}.
\end{equation}
A set $A$ is said to be {\it invariant} under the map $\phi^{(j)}$ if the set is 
unchanged when $\phi^{(j)}$ is applied to it, 
in which case from \eqref{Eq:Phi-j} and \eqref{Eq:DefnPhi} we would have that for each $a \in A$,  
\begin{align} \label{Eq:conseqOfInvariance}
\zeroAt{j}(a) \in A .
\end{align}
\begin{lemma} \label{Lem:InvariantSet}
For any $A \subseteq \{0, 1\}^k$ and all integers $m$ and $t$ such that $1 \le m \le t \le k$, 
the set $\phi^{(t)}(\phi^{(t-1)}(\cdots \phi^{(1)}(A)))$ is invariant under the map $\phi^{(m)}$.
\end{lemma}
\begin{proof}
For any $A' \subseteq \{0, 1\}^{k}$, we have 
%
\begin{equation}
\label{Eq:SelfInvariant}
\phi^{(i)}(\phi^{(i)}(A')) = \phi^{(i)}(A') \quad \forall \ i \in \{1,2,\ldots,k\} .
\end{equation}

The proof of the lemma is by induction on $t$. 
For the base case $t = 1$, 
the proof is clear since $\phi^{(1)}(\phi^{(1)}(A)) = \phi^{(1)}(A)$ from 
\eqref{Eq:SelfInvariant}. 
Now suppose the lemma is
true for all $t < \tau$ (where $\tau \ge 2$). 
Now suppose $t = \tau$. 
Let $B = \phi^{(\tau-1)}(\phi^{(\tau-2)}(\cdots \phi^{(1)}(A)))$.
Since $\phi^{(\tau)}(\phi^{(\tau)}(B)) = \phi^{(\tau)}(B)$ from \eqref{Eq:SelfInvariant}, 
the lemma is true when $m = t = \tau$. 
In the following arguments, we take $m < \tau$. 
From the induction hypothesis, $B$ is invariant under the map $\phi^{(m)}$, i.e., 
\begin{equation}
\label{Eq:IndHyp}
\phi^{(m)}(B) = B .
\end{equation}
Consider any vector $c \in \phi^{(\tau)}(B)$. 
From \eqref{Eq:conseqOfInvariance}, 
we need to show that $\zeroAt{m}(c)  \in  \phi^{(\tau)}(B)$. 
We have the following cases.
\begin{align}
\label{Eq:Ind1}
\VecComp{c}{\tau} = 1 \quad : \quad  & c,  \zeroAt{\tau}(c)  \in  B 
&\Comment{$\VecComp{c}{\tau} = 1$ and $c \in \phi^{(\tau)}(B)$} \ \ \\
\label{Eq:Ind2}
& \zeroAt{m}(c) \in B &\Comment{\eqref{Eq:IndHyp} and  \eqref{Eq:Ind1}} \\
\label{Eq:Ind3}
& \zeroAt{\tau}\left(\zeroAt{m}(c)\right) = \zeroAt{m}\left(\zeroAt{\tau}(c)\right) \in B &\Comment{\eqref{Eq:IndHyp} and \eqref{Eq:Ind1}} \\
\nonumber
& \zeroAt{m}(c)  \in  \phi^{(\tau)}(B) & \Comment{\eqref{Eq:Ind2} and \eqref{Eq:Ind3}}
\end{align}
\begin{align}
\label{Eq:Ind4}
\VecComp{c}{\tau} = 0 \quad : \quad & \exists \ b \in B \mbox{ such that } \zeroAt{\tau}(b) = c 
&\Comment{$\VecComp{c}{\tau} = 0$ and $c \in \phi^{(\tau)}(B)$} \\
\label{Eq:Ind5}
&\zeroAt{m}(b) \in B &\Comment{\eqref{Eq:IndHyp} and  \eqref{Eq:Ind4}} \\
\label{Eq:Ind6}
&\zeroAt{m}\left(\zeroAt{\tau}(b)\right) = \zeroAt{\tau}\left(\zeroAt{m}(b)\right) \in \phi^{(\tau)}(B) &\Comment{\eqref{Eq:Ind5}} \\
\nonumber
&\zeroAt{m}\left(c\right) \in \phi^{(\tau)}(B) 
&\Comment{\eqref{Eq:Ind4} and \eqref{Eq:Ind6}} .
\end{align}
Thus, the lemma is true for $t = \tau$ and the induction argument is complete.
\end{proof}
Let $A_1, A_2,\ldots,A_M \subseteq \{0, 1\}^{k}$ be such that $\card{A_{i}} \ge 2^{k - n}$ for each $i$. 
Let $\NbinaryBlocks{U} = \displaystyle \prod_{i=1}^{M} A_i$ 
and extend the definition of $\phi^{(j)}$ in \eqref{Eq:DefnPhi} to products by 
$$
\phi^{(j)}(\NbinaryBlocks{U})  = \displaystyle \prod_{i=1}^{M} \phi^{(j)}(A_i).
$$
$\NbinaryBlocks{U}$ is said to be \textit{invariant under $\phi^{(j)}$} if 
$$
\phi^{(j)}(\NbinaryBlocks{U}) = \NbinaryBlocks{U} .
$$
It can be verifed that $\NbinaryBlocks{U}$ is invariant under
$\phi^{(j)}$ iff each $A_i$ is invariant under $\phi^{(j)}$. 
For each $i \in \{1,2,\ldots,M\}$, let
$$
B_i = \phi^{(k)}(\phi^{(k-1)}(\cdots \phi^{(1)}(A_i))) 
$$
and from \eqref{Eq:TransformedCard} note that 
\begin{equation}
\label{Eq:CardConservation}
\card{B_{i}} = \card{A_{i}} \ge 2^{k - n}.
\end{equation}
Let 
$$
\NbinaryBlocks{V} =  \phi^{(k)}(\phi^{(k-1)}(\cdots \phi^{(1)}(\NbinaryBlocks{U}))) =  \displaystyle \prod_{i=1}^{M} B_i 
$$
and recall the definition of the function $Q$ \eqref{Eq:sumDefinition}. 
\begin{lemma}\label{Lemma1}
$$
\card{\sumset{\NbinaryBlocks{U}}} \ge \card{\sumset{\NbinaryBlocks{V}}} .
$$
\end{lemma}
\begin{proof}
We begin by showing that 
\begin{equation}
\label{Eq:Claim}
\card{\sumset{\NbinaryBlocks{U}}} \geq \card{\sumset{\phi^{(1)}(\NbinaryBlocks{U})}} .
\end{equation}
For every $p \in \{0, 1,\ldots, M\}^{k-1}$, 
let
\begin{eqnarray*}
\varphi(p) & = & \left\{ r \in \sumset{\NbinaryBlocks{U}} : \left( \VecComp{r}{2},\cdots,\VecComp{r}{k} \right) = p \right\}  \\
\varphi_1(p) & = & \left\{ s \in \sumset{\phi^{(1)}(\NbinaryBlocks{U})} : \left( \VecComp{s}{2},\cdots,\VecComp{s}{k} \right) = p \right\}
\end{eqnarray*}
and note that
\begin{eqnarray} \label{s}
\sumset{\NbinaryBlocks{U}} &=& \bigcup_{p \in \{0, 1,\ldots, M\}^{k-1}} \varphi(p) \\
\label{s1}
\sumset{\phi^{(1)}(\NbinaryBlocks{U})} &=& \bigcup_{p \in \{0, 1,\ldots, M\}^{k-1}} \varphi_1(p)
\end{eqnarray}
where the two unions are in fact disjoint unions. 
We show that for every $p \in  \{0, 1,\ldots, M\}^{k-1}$,
\begin{equation}
\label{SuffCond}
\card{\varphi(p)} \ge \card{\varphi_{1}(p)}
\end{equation}
which by (\ref{s}) and (\ref{s1}) implies \eqref{Eq:Claim}.

If $\card{\varphi_{1}(p)} = 0$, 
then (\ref{SuffCond}) is trivial. 
Now consider any $p \in  \{0, 1,\ldots, M\}^{k-1}$ such that $\card{\varphi_{1}(p)} \ge 1$ and let
$$
K_p = \max \left\{ i \ : (i, \VecComp{p}{1},\cdots,\VecComp{p}{k-1}) \in \varphi_{1}(p) \right\}.
$$
Then we have 
\begin{equation}
\label{Eq:UbNoOfSums}
\card{\varphi_{1}(p)} \leq K_p + 1. 
\end{equation}
Since $(K_p,\VecComp{p}{1},\cdots,\VecComp{p}{k-1}) \in \varphi_{1}(p)$, 
there exists $(a^{(1)}, a^{(2)},\ldots, a^{(M)}) \in \NbinaryBlocks{U}$ such that 
\begin{equation}
\label{Eq:a^i}
\sum_{i=1}^{M} \hat{\phi}^{(1)} \left(A_i,a^{(i)}\right) = 
(K_p,\VecComp{p}{1},\cdots,\VecComp{p}{k-1}).
\end{equation}
Then from the definition of the map $\hat{\phi}^{(1)}$ in \eqref{Eq:Phi-j},
there are $K_p$ of the $a^{(i)}$'s from amongst\\
$\{a^{(1)}, a^{(2)},\ldots, a^{(M)}\}$ such that $\VecComp{a^{(i)}}{1} = 1$ and
$\hat{\phi}^{(1)}\left(A_{i}, a^{(i)}\right) = a^{(i)}$. 
Let $I = \{i_1, i_2,\ldots, i_{K_p}\} \subseteq \{1,2,\ldots,M\}$ 
be the index set for these vectors and let $\hat{a}^{(i)} = \zeroAt{1}(a^{(i)})$ for each $i \in I$.
Then for each $i \in I$, we have  
\begin{align*}
a^{(i)} &= \left(1, \VecComp{a^{(i)}}{2},\dots,\VecComp{a^{(i)}}{k} \right) \ \in \ A_i \\
\hat{a}^{(i)} &= \left(0,\VecComp{a^{(i)}}{2},\dots,\VecComp{a^{(i)}}{k} \right) \ \in \ A_i 
        &\Comment{$\hat{\phi}^{(1)}\left(A_{i}, a^{(i)}\right) = a^{(i)}$ and  \eqref{Eq:Phi-j}} .
\end{align*}
Let 
\begin{equation}
\label{Eq:DefnR}
R = \left\{ \sum_{i=1}^{M} b^{(i)} \ : \ \begin{array}{ll} b^{(i)} \in \{a^{(i)},\hat{a}^{(i)}\}  
 &\mbox{for }  i \in I, \\ b^{(i)} = a^{(i)}  &\mbox{for }  i \notin I  \end{array}\right\} 
\subseteq \varphi(p).
\end{equation}
From \eqref{Eq:a^i} and \eqref{Eq:DefnR}, for every $r \in R$ we have
\begin{align*}
\VecComp{r}{1} &\in \left\{0,1,\ldots,\card{I}\right\}, \\
\VecComp{r}{i} &= \VecComp{p}{i}  \ \ \forall \ i \in \left\{2,3,\ldots,k\right\} 
\end{align*}
and thus
\begin{equation}
\label{Eq:CardR}
\card{R} = \card{I} + 1 = K_p + 1 .
\end{equation}
Hence, we have   
\begin{align*}
\card{\varphi(p)} &\geq \card{R}  &\Comment{\eqref{Eq:DefnR}}\\
&= K_p + 1 &\Comment{\eqref{Eq:CardR}} \\
&\geq \card{\varphi_{1}(p)} & \Comment{\eqref{Eq:UbNoOfSums}}
\end{align*}
and then from (\ref{s}) and (\ref{s1}), it follows that 
$$
\card{\sumset{\NbinaryBlocks{U}}} \ge \card{\sumset{\phi^{(1)}(\NbinaryBlocks{U})}} . 
$$
For any $A \subseteq \{0,1\}^{k}$ and any $j \in \{1,2,\ldots,k\}$, 
we know that $\card{\phi^{(j)}(A)} \subseteq \{0,1\}^{k}$. 
Thus, 
the same arguments as above can be repeated to show that 
\begin{align*}
\card{\sumset{\phi^{(1)}(\NbinaryBlocks{U})}} &\ge \card{\sumset{ \phi^{(2)}(\phi^{(1)}(\NbinaryBlocks{U}))}} \\
&\ge \card{\sumset{ \phi^{(3)}(\phi^{(2)}(\phi^{(1)}(\NbinaryBlocks{U})))}} \\
& \hspace{.5in} \vdots \\
&\ge \card{\sumset{ \phi^{(k)}(\phi^{(k-1)}(\cdots \phi^{(1)}(\NbinaryBlocks{U})))}} \\
& = \card{\sumset{\NbinaryBlocks{V}}} .
\end{align*}
\end{proof}
%
For any $s,r \in \integer^k$, 
we say that $s \le r$ if $\VecComp{s}{l} \le \VecComp{r}{l}$ 
for every $l \in \{1, 2,\ldots, k\}$. 
\NOPROCESS{ 
For a set $A$ containing elements from $\{0, 1\}^{k}$, 
we denote the maximum Hamming weight
amongst its elements by $\maxWeight{A}$. 
Recall that 
$$
\NbinaryBlocks{V} = \displaystyle \prod_{i=1}^{M} B_i =  \displaystyle \prod_{i=1}^{M} \phi^{(k)}(\phi^{(k-1)}(\cdots \phi^{(1)}(A_i))) .
$$
}
\begin{lemma} \label{Lem:majorization}
Let $p \in \sumset{\NbinaryBlocks{V}}$. 
If $q \ \in \ \{0, 1,\ldots, M\}^k$ and $q \le p$, 
then $q \in \sumset{\NbinaryBlocks{V}}$.
\end{lemma}
\begin{proof}
Since $q \le p$, 
it can be obtained by iteratively subtracting $1$ from the components of $p$, i.e., 
there exist $t \ge 0$ and $i_1,i_2,\ldots,i_t \in \{1,2,\ldots,k\}$ such that
$$
q = \zeroAt{i_1}\left(\zeroAt{i_2}\left(\cdots\left(\zeroAt{i_t}(p)\right) \right) \right) .
$$
Consider any $i \in \{1, 2,\ldots, k\}$. 
We show that $\zeroAt{i}(p) \in \sumset{\NbinaryBlocks{V}}$, 
which implies by induction that $q \in \sumset{\NbinaryBlocks{V}}$. 
If $\VecComp{p}{i} = 0$, then
$\zeroAt{i}(p) = p$ and we are done.
Suppose that $\VecComp{p}{i} > 0$.
Since $p \in \sumset{\NbinaryBlocks{V}}$,
there exists $b^{(j)} \in B_j$ 
for every $j \in \{1, 2,\ldots, M\}$ such that
$$
p = \sum_{j=1}^{M} b^{(j)}
$$
and $\VecComp{b^{(m)}}{i} = 1$ for some $m \in \{1, 2,\ldots,M\}$.
From Lemma~\ref{Lem:InvariantSet}, 
$\NbinaryBlocks{V}$ is invariant under $\phi^{(i)}$ and thus from 
\eqref{Eq:conseqOfInvariance}, $\zeroAt{i}(b^{(m)}) \in B_m$ and 
%
$$
\zeroAt{i}(p) = \sum_{j=1}^{m-1} b^{(j)} + \zeroAt{i}(b^{(m)}) + \!\!\sum_{j=m+1}^{M}\!\! b^{(j)} 
$$
is an element of $\sumset{\NbinaryBlocks{V}}$.
\end{proof}
The lemma below is presented in \cite{Ayaso07} without proof, as the proof is straightforward.
\begin{lemma} \label{Lem:optimization}
For all positive integers $k, n, M$, 
and $\delta \in (0,1)$, 
\begin{equation} \label{eq:minimiz}
 \min_{ \substack{ 0 \ \leq \ m_{i} \ \leq \ M, \\ 
 \sum_{i=1}^{k} m_{i} \ \geq \ \delta M k }} \ 
 \prod_{i=1}^{k}\left( 1 + m_{i} \right) \ge \left(M+1\right)^{\delta k } .
\end{equation}
\end{lemma}
For any $a \in \{0, 1\}^{k}$, 
let $\hammingWeight{a}$  denote the Hamming weight of $a$, i.e., 
the number of non-zero components of $a$. 
The next lemma uses the function $\gamma$ defined in \eqref{Eq:CardConst}.
\begin{lemma}\label{Lemma2}
$$
\card{\sumset{\NbinaryBlocks{V}}} \ge  (M + 1)^{\cardConst{k/n} k} .
$$
\end{lemma}
\begin{proof}
Let $\delta = \cardConst{k / n}$.
The number of distinct elements in $\{0, 1\}^{k}$ with Hamming weight at
most $\left\lfloor \delta k \right\rfloor$ equals
\begin{align*}
\sum_{j=0}^{\left\lfloor \delta k \right\rfloor } \binom{k}{j} &\leq 2^{k\entropy(\delta)} 
&\Comment{\cite[p.15, Theorem 1]{Hoeff63}}  \\
&= 2^{(k-n)/2} &\Comment{\eqref{Eq:CardConst}} .
\end{align*}
For each $i \in \{1, 2,\ldots, M\}$, $\card{B_i} \ge 2^{k-n}$ from 
\eqref{Eq:CardConservation} and hence there exists $b^{(i)} \in B_i$ such that
$\hammingWeight{b^{(i)}} \ge \delta k$. Let 
$$
p = \sum_{i=1}^{M} b^{(i)} \ \in \sumset{\NbinaryBlocks{V}} .
$$
It follows that $\VecComp{p}{j} \in \{0,1, 2,\ldots, M\}$ for every $j \in \{1, 2,\ldots, k\}$, and
\begin{equation}
\label{Eq:SumOfComp}
\sum_{j=1}^{k} \VecComp{p}{j} = \sum_{i=1}^{M} \hammingWeight{b^{(i)}} \ge 
\delta M k.
\end{equation}
%
The number of vectors $q$ in $\{0, 1,\ldots, M\}^k$ such that $q \preceq p$ equals 
$\displaystyle \prod_{j=1}^{k} \left( 1 + \VecComp{p}{j} \right)$, and from Lemma~\ref{Lem:majorization}, 
each such vector is also in $\sumset{\NbinaryBlocks{V}}$. 
Therefore, 
\begin{align*}
\card{\sumset{\NbinaryBlocks{V}}} &\ge \prod_{j=1}^{k} \left( 1 +
\VecComp{p}{j} \right) \\
&\ge \left(M+1\right)^{\delta k} &\Comment{\eqref{Eq:SumOfComp}
and Lemma~\ref{Lem:optimization}} .
\end{align*}
Since $\delta = \cardConst{k / n}$, 
the result follows. 
\end{proof}
\clearpage
\bibliographystyle{ieeeS}
\small
\bibliography{topology} 

\end{document}